\PassOptionsToPackage{dvipsnames}{xcolor}
\documentclass[letterpaper,11pt,USenglish]{scrartcl}
\usepackage[letterpaper,margin=1in]{geometry}

\usepackage{amsmath}
\usepackage{amsthm}
\usepackage{amssymb}

\usepackage{thmtools} 
\usepackage{thm-restate}

\usepackage{xspace}
\usepackage{cite}
\usepackage{hyperref}
\usepackage[capitalise]{cleveref}
\usepackage{xcolor}

\newtheorem{theorem}{Theorem}
\newtheorem{observation}[theorem]{Observation}
\newtheorem{definition}[theorem]{Definition}
\newtheorem{lemma}[theorem]{Lemma}

\crefname{observation}{Observation}{Observations}

\usepackage{tikz}
\usetikzlibrary{positioning, shapes.geometric, shapes, calc}

\DeclareMathOperator*{\polylog}{polylog}
\DeclareMathOperator*{\polyloglog}{polyloglog}

\DeclareMathOperator*{\outdeg}{out-deg}
\DeclareMathOperator*{\indeg}{in-deg}
\DeclareMathOperator*{\Bin}{Bin}
\DeclareMathOperator*{\Uniform}{Unif}

\newcommand{\case}[1]{\emph{Case $#1$.}}
\newcommand{\BigO}[1]{\ensuremath{O\left( #1 \right)}}

\newcommand{\Dist}[3]{\ensuremath{\mathcal{D}\left(#1,#2,#3\right)}}

\newcommand{\Glogt}{\ensuremath{\mathcal{G}_{\log}^t}}
\newcommand{\Glog}{\ensuremath{\mathcal{G}_{\log}}}
\newcommand{\Rootset}{\ensuremath{V}_{R}}
\newcommand{\Rootsett}{\ensuremath{V}_{R}^t}
\newcommand{\Ex}[1]{\ensuremath{\mathbb{E}\left[#1\right]}}
\newcommand{\Galpha}{\ensuremath{\mathcal{G}_{\alpha}^{\mathcal{F}}}}
\newcommand{\Ealpha}{\ensuremath{E_{\alpha}^{\mathcal{F}}}}

\makeatletter
\def\?#1{}
\def\whp{w.h.p\@ifnextchar.{.\whpFootnote/\?}{\@ifnextchar,{.\whpFootnote/}{\@ifnextchar){.\whpFootnote/}{\@ifnextchar:{.:\whpFootnote/\?}{.\whpFootnote/\ }}}}}
\def\Whp{W.h.p\@ifnextchar.{.\whpFootnote/\?}{\@ifnextchar,{\whpFootnote/.}{.\whpFootnote/\ }}}

\def\whpFootnote/{\global\def\whpFootnote/{}
}
\makeatother

\author{
Gregor Bankhamer\\\normalsize{Department of Computer Sciences}\\\normalsize{University of Salzburg, Austria}\\\normalsize{gbank@cs.sbg.ac.at} \and
Robert Elsässer\\\normalsize{Department of Computer Sciences}\\\normalsize{University of Salzburg, Austria}\\\normalsize{elsa@cs.sbg.ac.at} \and 
Stefan Schmid\\\normalsize{Technical University of Berlin, Germany and}\\\normalsize{Faculty of Computer Science, University of Vienna, Austria}\\\normalsize{stefan.schmid@tu-berlin.de} 
}

\title{Randomized Local Fast Rerouting for Datacenter\\Networks
with Almost Optimal Congestion\footnote{Research supported by the Vienna Science and Technology Fund (WWTF), project´ICT19-045 (WHATIF), 2020-2024.}}

\date{}

\begin{document}

\maketitle

\begin{abstract}
To ensure high availability, 
datacenter networks must rely on local fast rerouting mechanisms
that allow routers to quickly react to link failures, in a fully
decentralized manner. 
However, configuring these mechanisms to provide 
a high resilience against multiple failures while avoiding congestion
along failover routes is algorithmically challenging,
as the rerouting rules can only depend on local failure information and 
must be defined ahead of time.
This paper presents a randomized local fast rerouting algorithm for Clos networks,
the predominant datacenter topologies. 
Given a graph $G=(V,E)$ describing a Clos topology, 
our algorithm defines local routing rules for each node $v\in V$,
which only depend on the packet's destination and are conditioned on the incident link failures. 
We prove that as long as number of failures at each node does not exceed a certain bound, our algorithm achieves an asymptotically
minimal congestion up to $\polyloglog$ factors along failover paths.
Our lower bounds are developed under some natural 
routing assumptions.

\end{abstract}

\section{Introduction}

Due to the popularity of data-centric applications
and distributed machine learning,
datacenter networks have become a critical infrastructure of
the digital society.
To meet the resulting stringent dependability
requirements, datacenter networks 
implement fast failover mechanisms that enable routers
to react to link failures quickly and to reroute flows in a decentralized manner, 
relying on \emph{static} routing tables which include conditional \emph{local} failover rules. 
Such local failover mechanisms in the data plane can react to failures orders of magnitudes
faster than traditional global mechanisms 
in the control plane which, 
upon failure, may recompute routing tables by running the routing protocol again~\cite{frr-survey,francois2005achieving,liu2013ensuring}. 

Configuring fast failover mechanisms however is challenging under multiple link failures,
as the failover behavior needs to be pre-defined, before the actual
failures are known. In particular, rerouting decisions can only rely on local information, 
without knowledge of possible further failures downstream. Without precautions, 
a local failover mechanism may hence entail congestion or even 
forwarding loops, already under a small number of link failures, while these issues
could easily be avoided in a centralized setting.

More formally, resilience is achieved in two stages. 
First, we are given a graph $G=(V,E)$ describing an undirected network (without failures). 
Our task is to compute local failover rules for each node $v\in V$ which define for each packet 
arriving at $v$ to which incident link it should be forwarded, based on the packet's destination
(known as destination-based routing); 
these rules can be conditioned on the status of the links incident to $v$. 
Second, when an adversary fails multiple links in the network, packets are forwarded 
according to our pre-defined conditional rules. Our objective is to define the 
static routing tables (i.e., the \emph{rulesets}) in the first stage such that desirable properties are preserved in the second stage, 
in particular, connectivity and a minimal congestion.

This paper studies fast failover algorithms tailored
towards Clos topologies, and more specifically to \emph{(multirooted) fat-trees}~\cite{leiserson1985fat,al2008scalable,singh2015jupiter}, 
the predominant datacenter networks.
In particular, we consider a scenario where
$n-1$ sources inject one indefinite flow each to a single destination. This 
scenario has already been studied intensively in the literature~\cite{foerster2019casa,opodis13shoot,dsn17,icnp19}: it models practically important operations such as in-cast~\cite{wu2012ictcp,handley2017re}, and is also theoretically interesting as it describes a particularly challenging situation 
because it is focused on a single destination which can lead to bottlenecks.

The goal is to ensure that each flow reaches its destinations even in the presence of 
a large number of link failures while minimizing congestion: 
to provide a high availability, it is crucial to avoid a high load
(and hence delays and packet loss) on the failover paths.
The problem is related to classic load-balancing problems such as
balls-into-bins, however, our setting introduces additional dependencies
in that failover rules need to define valid paths.

\subsection{Our Contribution}

This paper studies the theoretical question of how to configure local fast
failover rules in Clos topologies such that connectivity is preserved and
load is minimized even under a large number of failures.

\paragraph{Results in a Nutshell:} We first derive a lower bound showing that 
by failing $\BigO{n / \log n}$ edges, the adversary can create a load of $\Omega(\log n / \log \log n)$ \whp in arbitrary topologies with $n$ nodes.
As a next step, we give a routing protocol for complete bipartite graphs that incorporates local failover rules and, for up to $\BigO{n / \log n}$ link failures,
achieves an almost minimal congestion (i.e., up to $(\log \log n)^2$ factor).
We then use the derived results to construct a failover ruleset for the Clos topology
with $L+1 = \Theta(1)$ levels and degree $k$ (cf~the definition in~\cref{sec:clos-definition} and the example in~\cref{fig:fattree-example}).
It is resilient to $\BigO{k / \log k}$ link failures, while keeping the total load below $\BigO{k^{L-1} \log k \cdot \log \log k}$ \whp. For a certain class of routing protocols, that only forward over shortest paths (according to the local view of the nodes, cf~\cref{def:fairly-balanced}) and exhibit a property we call \emph{fairly balanced}, this is again optimal up to $(\log \log n)^2$ factors. This class of protocols is natural and reminiscent of the widely-deployed shortest-path routing protocol ECMP (equal-cost multipath)~\cite{singh2015jupiter,kabbani2014flowbender}.

\paragraph{Techniques:}

In this work, we are interested in rulesets that include randomization, and are robust against an adversary, which knows the algorithm and the routing destination, but not the random choices leading to the specific failover routes.
In our lower bound analysis we need to cover a wide variety of failover protocols. While the deterministic case is well-understood \cite{opodis13shoot}, and failover protocols based on the uniform distribution are easy to handle, a mixture of the both is non-trivial to analyze.
We opt for a carefully crafted case-distinction 
that captures failover paths, which might be predicted by the adversary with good probability. 
For this, we exploit the properties of the subgraphs induced by the edges, which have a certain (high) probability to be chosen as failover links. In all the other cases, we are able to use the (not necessarily uniform) random placement of the loads initiated by a subset of source nodes, and apply a balls-into-bins style argument to show that at least one node will receive high load.
 
To develop an efficient protocol for the Clos topology, we exploit the fact that it contains multiple bipartite sub-graphs.
The main algorithm combines the advantages of deterministic protocols and forwarding loop-freeness, with the resilience of randomized approaches.
Our approach builds upon the Interval protocol in \cite{icnp19}, which is designed for the clique. However, the adaptation of this approach to the Clos topology comes with multiple challenges that need to be solved. The approach of \cite{icnp19} models the load that nodes receive with the help of trees that are tailored towards the clique, and this method can not be extended to more complex topologies. To overcome this problem, we use a Markov chain to 
describe such loads and develop a general Markov chain result that might be of independent interest (see \cref{thm:markovchain}). For Markov chains with state space $\mathbb{N}$ that drift towards $0$ and that
can be modelled with Poisson trials, it states a concentration inequality for the sum of the first $r = \Omega(1)$ elements.
Additionally, the protocol in \cite{icnp19} relies on splitting the nodes into partitions of similar size. Contrary to the clique, where the assignment of nodes to partitions can be arbitrary, this is challenging in the case of the Clos topology.
Furthermore, for our analysis in the Clos network we need to consider the flows arriving at a certain node from lower and upper levels concurrently.
This leads to dependencies, which prevents us from using standard techniques such as Chernoff bounds and the method of bounded differences. 
To overcome this problem, we uncover the failover edges step by step, and utilize an inductive approach over the increasingly small subtrees around the destination, bounding the number of flows entering the corresponding subtree. For the details see \cref{sec:clos-analysis}. 


\subsection{Related Work}

Motivated by measurement studies of network-layer failures in datacenters, 
showing that especially link failures are frequent
and disruptive \cite{gill2011understanding}, the problem of designing resilient routing
mechanisms has received much attention in the literature over the last years, see e.g.,~\cite{chiesa2017ton,robroute16infocom,icalp16,apocs21resilience,podc-ba,opodis13shoot} or the recent survey by Chiesa et al.~\cite{frr-survey}.

In this paper, we focus on the important model in which we do not allow for packet header rewriting
or maintain state at the routers, which rules out approaches such as link reversal and
others~\cite{gafni-lr,opodis20}. 
A price of locality for static rerouting mechanisms has first been shown by Feigenbaum et al.~\cite{podc-ba} and
Borokhovich et al.~\cite{opodis13shoot}, who proved that
it is not always possible to locally reroute packets to their destination even if the underlying network remains connected after the failures. These impossibility results have recently
been generalized to entire graph families by Foerster et al.~\cite{apocs21resilience}.
On the positive side, 
Chiesa et al.~showed that highly resilient failover algorithms can be realized
based on arc-disjoint arborescence covers~\cite{chiesa2017ton,robroute16infocom,icalp16},
an approach  which generalizes traditional solutions based on spanning trees \cite{tapolcai2013sufficient}. 
However, these papers only focus on connectivity and do not consider congestion on the resulting
failover paths.
Furthermore, while arborescence-based approaches have the advantage that they are naturally loop-free,
they result in long paths (and hence likely high load) and are complex to compute. 

Only little is known about local fast failover
mechanisms that account for load.
Pignolet et al.~in~\cite{dsn17} showed that 
when only relying on deterministic destination-based failover rules, 
an adversary can always induce a maximum edge load of $\Omega(\varphi)$ 
by cleverly failing $\varphi$ edges; when failover rules can also depend on the 
source address, an edge load of $\Omega(\sqrt{\varphi})$ 
can still be generated, when failing $\varphi$ many edges. 
In \cite{foerster2019casa},
Foerster et al.~build upon \cite{opodis13shoot,dsn17},
and leverage a connection to distributed computing problems without communication \cite{Malewicz2005},
to devise a fast failover algorithm which balances load across arborescences using
combinatorial designs.
In these papers the focus is on deterministic algorithms.

Our work builds upon \cite{icnp19} where we showed that randomized algorithms
can reduce congestion significantly in complete graphs. In particular, we presented 
three failover strategies:  
Assuming up to $\varphi= O(n)$ edge failures, 
the first algorithm ensures that w.h.p.~a load of $O(\log n \log \log n)$ 
will not be exceeded at most nodes, while the remaining $O(\polylog n)$ nodes reach a load of
at most $O(\polylog n)$. 
The second approach reduces the edge failure resilience to $O(n/\log n)$
but only requires knowledge of the packet destinations,
and achieves a congestion of only $O(\log n \log \log n)$ at any node w.h.p.
Finally, by assuming that the nodes do have access to $O(\log n)$ shared permutations of $V$, which are not known to the adversary, the node load can be reduced even further: 
a maximum load of only $O(\sqrt{\log n})$ occurs at \emph{any} node w.h.p. 
However, our work relied on the assumption that the underlying network is fully 
connected (i.e., forms a clique).
That said, our simulations (performed after we published that paper) showed
promising first results for interval-based routing on Clos topologies as well.

In this work, we consider randomized fast failover specifically in the context 
of datacenter networks which typically rely on Clos topologies (also known as
multi-rooted fat-trees)~\cite{leiserson1985fat,al2008scalable,singh2015jupiter}.
This scenario is not only of practical importance, but also significantly more challenging.
Nevertheless, we are able to derive almost tight upper and lower bounds for this setting,
under some natural fairness and shortest path assumptions.

\begin{figure}
\centering
\scalebox{0.8}{
 \input{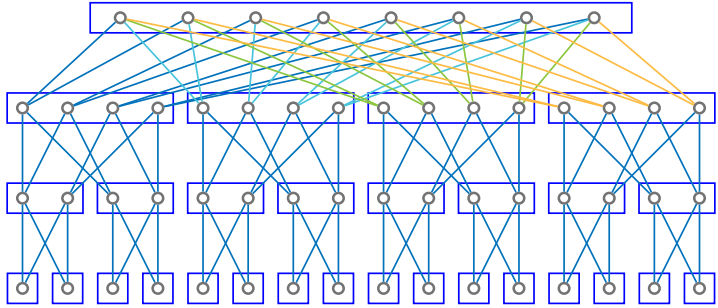}
}
\caption{Clos topology with levels $0,1,2,3$ (numbered from top to bottom)
and degree $k=4$} \label{fig:fattree-example}
\vspace{-.3cm}
\end{figure}

\subsection{Model}

In a nutshell, our model includes two stages.
First, we are asked to define the rulesets of the (static) routing tables
of each node $v$ in the network; these rules can depend on the destination
and be conditioned on the possible link failures incident to $v$ (i.e., only
the \emph{local} failures). 
Later, an adversary will decide which links to fail in the network; as
the routing tables defined before are static and cannot be changed depending on the
actual failures, packets will now simply be forwarded according to the local 
failover rules. Our objective is to pre-define these ruleset such that routing 
reachability is preserved under these failures and the load is minimized.

\label{sec:local-failover} \label{sec:model}
\paragraph{Local Destination-Based Failover Routing:} We represent our network as an undirected graph $\mathcal{G} = (V,E)$ and denote by $\mathcal{F}$ the set of failed edges. Each node $v$ is equipped with a static routing table $\alpha$, which we assume to be precomputed without knowledge of $\mathcal{F}$.
When a packet with destination $d$ arrives at a node $v$, it is forwarded to the neighbor of $v$ specified in the routing entry $\alpha(v, \mathcal{F}_v, d)$. Here $\mathcal{F}_v = \{w ~|~ (v,w) \in \mathcal{F}\}$ denotes the set of unreachable neighbors of $v$ (which may be empty). 
In order to allow for randomization we assume that, for each node $v$, the entry $\alpha(v, \mathcal{F}_v, d)$ is drawn from $\Dist{v}{\mathcal{F}_v}{d}$, which is a distribution over $(\Gamma(v) \setminus \mathcal{F}_v) \cup \{v\}$. Here $\Gamma(v)$ denotes the set of neighbors of $v$.  This way, a local failover routing protocol $\mathcal{P}$ can be described by the set of its distributions $\mathcal{P} = \{ \Dist{v}{\mathcal{F}_v}{d} ~|~ v \in V, \mathcal{F}_v \subseteq \Gamma(v), d \in V\}$. We call such a protocol $\mathcal{P}$  \emph{destination-based} as the only header information that is used for forwarding decisions is the destination address. In the following, we will also assume that for every $v$ the edge $(v,v)$ exists in $\mathcal{G}$ without being included in $\Gamma(v)$. This looping edge cannot be failed by the adversary and is used to enable the analysis of the case where $\mathcal{F}_v  = \Gamma(v)$.
Throughout the following sections, we assume the existence of an adversary, which knows the employed protocol, or in other words, the set of distributions $\mathcal{P}$ and may construct the set $\mathcal{F}$. However, this adversary does not know the random choices that lead to the routing entries $\alpha$.
In practice, to hide from the adversary longer term, this could for example be realized by generating new random tables once in a while.

\paragraph{Traffic Pattern and Load:}
Our focus lies on \emph{flow}-based \emph{all-to-one} routing~\cite{foerster2019casa,dsn17}. That is, we assume that every node $v \in V \setminus \{d\}$ sends out an indefinite flow of packets towards some common destination $d$. For a node $v$, we will then say that it has $\emph{load}$ of $\ell$ (or short: $\mathcal{L}(v) = \ell)$ iff $\ell$ such flows cross node $v$ on their way to destination $d$. Similarly, for an edge $e$ we define $\mathcal{L}(e)$ to denote the number of flows forwarded over edge $e$. In case some flow travels in a forwarding cycle, we say that all edges and nodes that lie on this cycle have infinite load.

Observe that, because we consider a purely destination-based ruleset, two flows hitting some node $w$ will be forwarded via the same routing entry $\alpha(w,\mathcal{F}_w,d)$. Therefore, as soon as flows hit the same node they cannot be separated anymore.
Note that this implies that $\max_{v \in V \setminus \{d\}} \{ \mathcal{L}(v) \} = \max_{e \in E} \{ \mathcal{L}(e) \}$ as the node $v$ with maximum load in $V \setminus \{d\}$ needs to forward all its flows over the edge $(v,d)$ to reach the destination.
As the edge load can be inferred from the node load, we will in the remaining part of our work only consider node loads.

\subsection{Conventions and Structure}

In the remainder of the paper, when we say we apply Chernoff bounds, we mean the usual multiplicative version (stated in \cref{thm:chernoff} for convenience). Furthermore, we denote by $\Bin(n,p)$ the binomial distribution with $n$ trials and success probability $p$, and by $\Uniform(A)$ the uniform distribution over elements in the set $A$. Finally, we denote by \whp ("with high probability") a probability of at least $1-n^{-\Omega(1)}$, where $n$ is the networks size.

We start by presenting a lower bound in \cref{sec:lower-bound}, stating that there exists a set of edge failures which induces a high load for any network and local destination-based failover protocol $\mathcal{P}$.
In \cref{sec:bipartite} we present an efficient loop-free failover protocol, which operates in complete bipartite graphs.
This is used as a preliminary result to develop a protocol in \cref{sec:fat-tree}, which may be employed in Clos topologies \cite{al2008scalable}. Each such section comes with a dedicated theorem and the corresponding analysis.
Certain technical details as well as concentration inequalities required in the analysis of \cref{sec:lower-bound,sec:bipartite,sec:fat-tree} are given in the technical details section (\cref{sec:details}). Finally, we discuss future research in \cref{sec:conclusion}.

\section{Lower Bound for Local Destination-Based Failover Protocols}
\label{sec:lower-bound}

In the following section we construct a lower bound, stating that by failing $\BigO{n / \log n}$ edges the adversary can w.h.p.~always create a load of $\Omega(\log n / \log \log n)$. 

\begin{theorem}
\label{thm:lower-bound-main}
    Consider any local destination-based  failover protocol $\mathcal{P}$ that operates in a graph $\mathcal{G} = (V, E)$ with $|V| = n$ and assume that all nodes perform all-to-one routing to some node $d \in V$. Then, if $d$ and $\mathcal{P}$ are known, a set of failures $\mathcal{F} \subseteq  \{ (v,d) ~|~ v \in V\}$ with $|\mathcal{F}| = \BigO{ n/ \log n}$ can be constructed such that some node $v \neq d$ has $\mathcal{L}(v) > (1/10) \cdot  \log n / \log \log n$ \whp. 
\end{theorem}

Most parts of our analysis are concerned with showing that the lower bound in \cref{thm:lower-bound-main} holds if $\mathcal{G}$ is the complete graph.
Therefore, the following notation will be defined with this constraint in mind. We focus on an arbitrary but fixed destination-based failover protocol $\mathcal{P} = \{\Dist{v}{\mathcal{F}_v}{d} ~|~ v,d \in V, \mathcal{F}_v \subseteq V \setminus \{v\}\}$ and also fix the destination node $d \in V$.
We will construct a set of failures $\mathcal{F}$ that induces a high load \whp by only failing edges of the form $(v,d)$, i.e., edges incident to the destination. This set of failures will have size $|\mathcal{F}| \leq \varepsilon \cdot n/\log n$, where $\varepsilon > 0$ is an arbitrary small constant.
In this setting $\mathcal{F}_v$ -- the set of unreachable neighbors of each node $v$ -- must either be $\{d\}$ or $\emptyset$ for any node $v$.
We then abbreviate $\alpha(v, \{d\},d)$ as $\alpha(v)$, which is the node to which $v$ forwards its load in case the link $(v,d)$ is failed.
Similarly, we abbreviate the corresponding probability distribution $\mathcal{D}(v, \{d\},d)$ as $\mathcal{D}_v$ and define $f_v$ to denote the probability density function (PDF) of $\mathcal{D}_v$.

\begin{definition}[Load Graphs]
\label{def:glog}
\label{def:glogt}
	The following directed graphs lie at the core of our analysis.
	\begin{enumerate}
		\item $\Galpha  = (V \setminus \{d\},\Ealpha)$ where $\Ealpha = \{ (v , \alpha(v, \mathcal{F}_v, d)) ~|~ v \in V \setminus \{d\} \}$ 
		\item  $\Glog = (V \setminus \{d\}, E_{\log})$ a directed graph with and  $E_{\log} := \{ (v,w) ~|~ v \in V \setminus \{d\} \land  f_v(w) > 1/ \log^4 n\}$
		\item $\Glogt$ which is constructed from $\Glog$ by removing edges in the following way:
		\begin{enumerate}
			\item First, remove arbitrary (outgoing) edges from nodes $v$ with $\outdeg(v) > 1$ until every node has degree $\leq 1$.
			\item Second, break any remaining cycle by removing an arbitrary edge from each cycle.
		\end{enumerate}
	\end{enumerate}
\end{definition}

The graph $\Galpha$, given a set of failures $\mathcal{F}$ and destination $d$, describes the path that flows take.  In order to fulfill our goal of creating a high load at some node $v$, we will make sure that some node $v$ is reached by many nodes in $\Galpha$. Note that this graph is a random variable as the entries in $\alpha$ follow distributions.

\begin{observation}
\label{obs:alpha}
\label{obs:load}
	 $\mathcal{L}(v) = \infty$ iff $v$ lies on a cycle in $\Galpha$. Otherwise $\mathcal{L}(v)$ is equal to the number of nodes $w \in V$ such that a path from $w$ to $v$ exists in $\Galpha$.
\end{observation}

The graph $\Glog$ allows us to capture whether the protocol $\mathcal{P}$ contains many failover edges that may be predicted by the adversary. Note, if the edge $(v,w) \in E_{\log}$ and $(v,d)$ are failed, then $v$ forwards its flows to $w$ with probability larger than $1/\log^4 n$.  Finally, $\Glogt$ is just a subgraph of $\Glog$, which does not contain any cycles and simplifies our analysis in some cases. These graphs are related to $\Galpha$ in the following way.

\begin{observation}
\label{obs:edge-alpha}
\label{obs:edge-alpha-t}
\label{obs:pdf}
	If $\mathcal{F}_v = \{d\}$, then $(v,w) \in \Galpha$ with probability $f_v(w)$. This probability is independent of other edges $(s,t)$ with $s \neq v$ being in $\Galpha$. In case $(v,w) \in \Glog$ (or $(v,w) \in \Glogt$) it follows that $f_v(w) > 1/\log^4 n$. 
\end{observation}

Intuitively, if $\Glog$ contains many edges, then we are in a setting close to determinstic failover protocols. By carefully failing edges of the form $(v,d)$, we have a good chance to make them appear in $\Galpha$ and create a node $v$ which is reached by many other nodes.
One final definition involves the natural definition of a reverse tree:
it is reversed in the sense that all edges are oriented towards the root.

\begin{definition}[Reverse Tree]
\label{def:trees}
   We call a directed graph $\mathcal{G}$ \emph{reverse tree} iff 
   \begin{enumerate}
       \item there is a node $r$ in $\mathcal{G}$ with $\outdeg(r) = 0$ that can be reached from all nodes in $\mathcal{G}$, and 
       \item every node $v \neq r$ in $\mathcal{G}$ has $\outdeg(v) = 1$.
   \end{enumerate}
   We call $r$ the \emph{reverse root} of $\mathcal{G}$. Furthermore, we call a graph $\mathcal{G}'$  \emph{reverse subtree} of $\mathcal{G}$ iff $\mathcal{G}'$ is both, a reverse tree and a subgraph of $\mathcal{G}$.
\end{definition}

Note, from the construction of $\Glogt$ it follows that it is a reverse forest.
Let the sets $\Rootset$ and $\Rootsett$ contain the nodes of out-degree $0$ in $\Glog$ and $\Glogt$, respectively.
\cref{obs:edge-alpha} implies for nodes in $\Rootset$ that we cannot easily predict their forwarding targets. Additionally, when constructing $\Glogt$ (see \cref{def:glogt}) in the first step, not a single node has its out-degree modified to $0$.  This can only happen in the second step, where exactly one node turns into a reverse root. Therefore, the difference $|\Rootsett| - |\Rootset|$ is equal to the number of cycles that were removed in this second step. 

\paragraph{Analysis Outline:} In the following subsections of the analysis we focus on the complete graph.
Depending on the structure of $\Glog$ and $\Glogt$ we split our analysis into 3 cases. 
As this graphs are inferred from $\mathcal{P}$, this can also be seen as a distinction between different types of routing protocols. In \cref{sec:lb-complete-case1}, we consider the case $|\Rootsett| - |\Rootset| \geq \sqrt{n}$, which intuitively corresponds to the case where $\mathcal{P}$ is prone to produce forwarding loops. In the second case (\cref{sec:lb-complete-case2}), we consider $|\Rootset| \geq \varepsilon n/\log n$, which implies that there are many nodes of degree $0$ in $\Glog$. Such nodes do not have a preferred forwarding target in case their link to $d$ is failed. They behave similarly to nodes that forward their flows to neighbors selected uniformly at random.
In the last case (\cref{sec:lb-complete-case3}), we consider $|\Rootset| \leq \varepsilon n/\log n$. In this case most nodes have at least one out-going edge in $\Glog$, which can be exploited by the adversary. In any of the three cases, we show that, by failing at most $\BigO{n/ \log n}$ edges, a  load of $(1/10) \cdot \log n / \log \log n$ is accumulated at some node in the network w.h.p.
Finally, we give the proof of \cref{thm:lower-bound-main} in \cref{sec:lower-theorem-proof}.

\subsection{\texorpdfstring{Analysis Case 1: $|\Rootsett| - |\Rootset| \geq \sqrt{n}$}{Analysis Case 1}}
\label{sec:lb-complete-case1}
Recall, the condition of this case implies that $\Glog$ contains at least $n^{\varepsilon}$ many cycles. The idea is, to fail the edge $(v,d)$ for many nodes $v$ that lie on such a cycle. Then, either the whole cycle or at least a long path of nodes lying on such a cycle appears in $\Galpha$ \whp and causes high load. We present the detailed proof in \cref{sec:details-lower-bound} on page \pageref{proof:lb-case1}. 

\begin{restatable}{lemma}{lbcaseone}
\label{lem:lb-case1}
	There exists a set of failures $\mathcal{F} \subseteq \{ (v,d) ~|~ $v$ \text{ lies on a cycle in } \Glog\}$ with $|\mathcal{F}| = \sqrt{n} \cdot (1/10) \cdot \log n / \log \log n$ such that some node $v$ that lies on a cycle has $\mathcal{L}(v) \geq (1/10) \cdot \log n / \log \log n$ \whp.
\end{restatable}

\subsection{\texorpdfstring{Analysis Case 2: $|\Rootset| \geq \varepsilon n/\log n$}{Analysis Case 2}}
\label{sec:lb-complete-case2}
In this setting, many nodes $v$ have out-degree $0$ in $\Glog$. In case the link $(v,d)$ of such a node is failed, it is hard to predict the failover edge $(v,\alpha(v))$ as these nodes have multiple potential forwarding targets. However, this can be exploited as there must be a set of nodes which are potential forwarding targets of many nodes in $\Rootset$. Similarly as in the analysis of a balls-into-bins process, we deduce that, \whp, there is one such node that receives load from $\Omega(\log n \cdot \log \log n)$ nodes in $\Rootset$. 
To simplify our analysis, we let $\Rootset'$ be an arbitrary but fixed subset of $\Rootset$ with size exactly $\varepsilon n/ \log n$. The proof for the following statement is given in \cref{sec:details-lower-bound} on page \pageref{proof:case2-1}. Note that, if the second statement of the following lemma holds, we are already done. 

\begin{restatable}{lemma}{lbcaseonetwo}
\label{lem:case-2-1}
	For $\mathcal{F} = \{ (v,d) ~|~ v \in \Rootset'\}$ one of the following statements holds:
	\begin{enumerate}
		\item In expectation, at least $n^{1/8}$ many nodes $w \in V$ have each at least $(1/10) \cdot  \log  n/ \log \log n$ incident edges that originate from $V_R'$.
		\item \Whp, there exists a node $w$ with $\mathcal{L}(w) = \log^2 n (1 - o(1))$.
	\end{enumerate}
\end{restatable}
\vspace{0.2cm}
The following statement can be shown with the help of standard-techniques, namely the \emph{method of bounded differences} (see  \cref{thm:mcdiarmid} on page  \pageref{thm:mcdiarmid}-- c.f \cite{MRT12}). The detailed proof is given in \cref{sec:details-lower-bound} on page \pageref{proof:case2-2}.

\begin{restatable}{lemma}{lbcasetwotwo}
\label{lem:case2-2}
	 Let $\mathcal{F} = \{ (v,d) ~|~ v \in \Rootset'\}$ and assume that the first statement of \cref{lem:case-2-1} holds. Then, \whp, there exists a node $v$ such that $\mathcal{L}(v) > (1/10) \cdot  \log n / \log \log n$.
\end{restatable}

\subsection{\texorpdfstring{Analysis Case 3: $|\Rootset| <  \varepsilon n/ \log n$}{Analysis Case 3}}
\label{sec:lb-complete-case3}

We assume throughout this section that the condition $|\Rootsett| - |\Rootset| \geq \sqrt{n}$ does not hold as that case was already analysed in \cref{sec:lb-complete-case1}. This, however, implies that $|\Rootsett| < |\Rootset|  +  \sqrt{n} < 2 \varepsilon n/ \log n$. Recall \cref{def:trees} and  that $| \Rootsett|$ is the number of reverse roots; or in other words, the number of reverse trees in the forest $\Glogt$. 
The idea behind this section is simple. If, for the nodes of some reverse subtree in $\Glogt$ of size $(1/10) \cdot \log n \cdot \log \log n$, we fail the edges incident to destination $d$, then this whole subtree will appear in $\Galpha$ with probability $\log^{-(4/10) \log n / \log \log n} n = n^{-4/10}$. By \cref{obs:load} the root of this tree will then receive $(1/10) \cdot \log n / \log \log n$ load. The main challenge is to construct a large enough set of independent trees such that at least one of them appears in $\Galpha$ \whp.
This is where the following counting argument comes into play.
\begin{observation}
\label{obs:case3}
	If  $|\Rootset| <  \varepsilon n/ \log n$ and $|\Rootsett| - |\Rootset| < \sqrt{n}$ then $\Glogt$ must contain one of the following
	\begin{enumerate}
		\item $\sqrt{n}$ disjoint reverse trees of size $> (1/10) \cdot \log n / \log \log n$ each, or
		\item one reverse tree of size > $\sqrt{n} / 2$.
	\end{enumerate}
\end{observation}
\begin{proof}
The proof follows by a counting argument. Assume both statements do not hold. Then, the  number of nodes $\Glogt$ contains can be upper-bounded by
\[
	 \left( \frac{2\varepsilon n}{ \log n} - \sqrt{n} \right) \cdot (1/10) \frac{\log n}{\log \log n} +  \sqrt{n} \cdot \frac{\sqrt{n}}{2} < n-1.
\]
The first product reflects that all but $\sqrt{n}$ trees have size at most $(1/10) \cdot \log n / \log \log n$. The second product reflects the worst-case of each of these at most $\sqrt{n}$  remaining trees having size $\sqrt{n} / 2$.
The above inequality chain leads to a contradiction as $\Glogt$ contains $n-1$ nodes.
\end{proof}

We now present a lemma for both of the cases in \cref{obs:case3}, each achieving the lower bound in \cref{thm:lower-bound-main}. The proofs follow the ideas sketched at the start of this section.
In case of \cref{lem:case3-2}, the tree of size $\geq \sqrt{n}/2$ needs to be split into $\sqrt{n} / \polylog n$ node-disjoint subtrees of size $(1/10) \cdot \log n / \log \log n$.
The proofs are given in \cref{sec:details-lower-bound}.

\begin{restatable}{lemma}{lemcasethreeone}
\label{lem:case3-1}
	Assume there are $\sqrt{n}$ reverse trees in $\Glogt$ of size at least $(1/10) \cdot \log n / \log \log n$ each. Then, there exists a failure set $\mathcal{F} \subseteq \{ (v,d) ~|~ v \text{ lies on a tree in $\Glogt$}\}$ with $|\mathcal{F}| = \sqrt{n} \cdot (1/10)  \cdot \log n / \log \log n$ such that a node $v$ has $\mathcal{L}(v) > (1/10)  \cdot \log n / \log \log n$.
\end{restatable}

\begin{restatable}{lemma}{lemcasethreetwo}
\label{lem:case3-2}
	Assume there is a reverse tree $\mathcal{T}_R$  of size $\sqrt{n} / 2$ in $\Glogt$. Then, there exists a set of failures $\mathcal{F} \subseteq \{ (v,d) ~|~ v \text{ lies in } \mathcal{T}_R\}$ with $|\mathcal{F}| \leq \sqrt{n} \cdot (1/10) \cdot  \log n / \log \log n$ such that a node $v$ of $\mathcal{T}_R$ has $\mathcal{L}(v) > (1/10) \cdot  \log n / \log \log n$ \whp
\end{restatable}

\subsection{\texorpdfstring{Proof of \cref{thm:lower-bound-main}}{Proof of Theorem}}
\label{sec:lower-theorem-proof}

    In \cref{sec:lb-complete-case1,sec:lb-complete-case2,sec:lb-complete-case3} we considered the complete graph $\mathcal{G}$ together with a fixed destination-based protocol $\mathcal{P}$ and all-to-one destination $d$. In this setting, we constructed the graphs $\Glog, \Glogt$ and split our analysis into three cases, depending on the structure of these graphs. In each case, we establish that the theorem's result w.r.t. $\mathcal{G}$ holds:
    \begin{enumerate}
        \item Case 1: If the case in \cref{sec:lb-complete-case1} occurs, then the result immediately follows from \cref{lem:lb-case1}.
        \item Case 2: If we are in the case of \cref{sec:lb-complete-case2}, then either the second statement of \cref{lem:case-2-1} holds and the result follows, otherwise the first statement holds and \cref{lem:case2-2} leads to the desired result.
        \item Case 3: This case was covered in \cref{sec:lb-complete-case3}, and further splits into two sub-cases as indicated in \cref{obs:case3}. In both sub-cases, the result follows as stated in \cref{lem:case3-1,lem:case3-2}, respectively. 
    \end{enumerate}
The proof for general undirected graphs $\mathcal{G} = (V,E)$ with $|V| = n$ follows from \cref{lem:general-graph-extension}, which we state in \cref{sec:details-lower-bound} on page  \pageref{lem:general-graph-extension}. The basic idea is that, for any protocol $\mathcal{P}$ operating in $\mathcal{G}$, one can construct an equivalent protocol $\mathcal{P}_K$ that operates in the clique $K_n$ (equivalent in the sense that the path flows take is the same in both graphs). We then use the statement of \cref{thm:lower-bound-main}, which we already established for complete graphs, to deduce that a set of failures $\mathcal{F}^{(K)}$ exists that induces a high load in $K_n$. The same set of failures (excluding some edges which may not exist in $\mathcal{G}$) also leads to a high load in $\mathcal{G}$ when employing $\mathcal{P}$. \qedhere

\section{Interval Routing in the Bipartite Graph}
\label{sec:bipartite}

In the following section, we will construct an efficient local failover protocol for the complete bipartite graph $G = (V \cup W, E)$. Here the set of nodes $V \cup W$ consists of two sets, where $|V| = |W| = n$ and edges are drawn such that each node $v \in V$ is connected to every $w \in W$ and vice versa.

	To employ our routing protocol, we further assume that the nodes in  both, $V$ and $W$, are partitioned into $K := C \log n$ sets, where $C = \Theta(1)$ is an arbitrary value larger $4$. That is, $V = V(0) \cup V(2) \cup ... \cup V(K-1)$ and $W = W(0) \cup W(2) ... \cup  W(K-1)$, where we assume that all these partitions have size $I := n/K = n/ (C \log n)$ (assume $C \log n$ divides $n$). We propose the following local routing protocol, which is resilient to  $\Omega(n / \log n)$ edge failures.
	\begin{definition}[Bipartite Interval Routing]
		We define the routing protocol $\mathcal{P}_{B}$, induced by the following distributions when routing towards some node $d \in W$ 
		\begin{itemize}
			\item For $v \in V(i)$ we set $\Dist{v}{\mathcal{F}_v}{d} = \Uniform(\{d\})$ if $d \not \in \mathcal{F}_v$, otherwise $\Dist{v}{\mathcal{F}_v}{d} = \Uniform( W(i) \setminus \mathcal{F}_v)$.
			\item For $w \in W(i)$ with $w \neq d$ we set $\Dist{w}{F_w}{d} = \Uniform(V((i+1) \mod K) \setminus \mathcal{F}_w)$.
		\end{itemize}
	\end{definition}

	 Note that this protocol is inspired by the \emph{Interval} routing protocol of \cite{icnp19} which is constrained to complete graphs. Intuitively, a packet with source in the set $V(i)$ follows the partitions $V(i) \rightarrow W(i) \rightarrow V(i+1) \rightarrow W(i+1) ... $ until reaching a node $v \in V$ such that $(v,d)$ is not failed. Therefore, the only way for flows to end up in a cycle is by travelling through all $2K$ intervals, which is very unlikely. We may also refer to this alternation between layers $V$ and $W$ of a packet as "ping-pong" in the remainder of the paper. 

	\begin{theorem}
	\label{thm:bipartite}
		Let $G = (V \cup W, E)$ be a complete bipartite graph with $|V| = |W| = n$. Let the routing protocol $\mathcal{P}_{B}$ be employed, configured with $C > 4$, and  all-to-one routing towards some destination $d \in W$ be performed.  Assume the set of failures $\mathcal{F}$ fulfills for every $i$ with $0 \leq i < K$ that 
		\begin{enumerate}
		    \item  $\forall w \in W: |\{ v \in V(i) ~|~ w \in \mathcal{F}_v\}| \leq I/3$, and 
		    \item $\forall v \in V:  |\{ w \in W(i) ~|~ v\in \mathcal{F}_w\}| \leq I/3$.
		\end{enumerate}
		Then, with probability at least $1-3n^{-(C-1)}$, every node $u \in V \cup W$ with $u \neq d$ has $\mathcal{L}(u)  = \BigO{\log n \cdot \log \log n}$, even if $\mathcal{F}$ is constructed with knowledge of $\mathcal{P}_B$ and $d$.
	\end{theorem}
	
	Intuitively, the constraint on $\mathcal{F}$ states that at most a $(1/3)$ fraction of nodes in the same interval may have failed edges incident to the same node. The constraint is, for example, easily fulfilled in case only $I/3 = n/(3C \log n) = \Theta(n / \log n)$ edges are failed in total. This implies that the load induced by the protocol approaches the lower bound in \cref{thm:lower-bound-main} up to only a $\text{polyloglog } n$ factor. In any deterministic protocol, a load of $\Omega(n / \log n)$ could be created in this setting \cite{opodis13shoot}. Additionally, the simple randomized protocol, which forwards the packets between nodes of $V$ and $W$ which are selected uniformly at random until a node $v \in V$ is reached such that $(v,d)$ is not failed, is prone to cycles. By failing $\BigO{n/\log n}$ arbitrary edges between nodes in $V$ and $d$, at least one flow will travel from such a node $v \in V$ to some $w \in W$ and back to $v$ with probability $\geq 1/\polylog n$. This creates a forwarding loop of length $2$ and prevents some flows from reaching destination $d$. Our interval protocol is hybrid in the sense that nodes forward their packet uniformly at random according to pre-determined partitions. This allows it to keep the network load low while also avoiding forwarding loops \whp.
	
	\subsection{Analysis of the Bipartite Interval Protocol}
	
	We consider a fixed destination node $d$ together with a set of failures $\mathcal{F}$ that fulfills the requirements of \cref{thm:bipartite}. To make our analysis more readable, we assume that all partitions $V(i)$ and $W(j)$ have exactly the same size. Furthermore, we denote by $\alpha(v)$ the (random) node that $v$ forwards packets towards destination $d$ when following $\mathcal{P}_B$. Before starting with the proof of the theorem, we show the following important statement, which implies that, \whp, no flows travel in a cycle until they reach the destination $d$. Packets "ping-pong" between nodes in $V$ and $W$ until they reach the destination. Due to the restrictions on the failure set in \cref{thm:bipartite} it follows that each time a packet lands on some node $v \in V$, there is a constant probability that the link $(v,d)$ is not failed. It follows that, \whp, the packet reaches $d$ after $K = \Theta(\log n)$ alternations between $V$ and $W$. A detailed proof is given in \cref{sec:details-bipartite} on page  \pageref{proof:bipartite-loop}.
	
	\begin{restatable}{lemma}{bipartitenoloop}
	\label{obs:bip-noloops} \label{lem:bip-noloops}
		Any packet starting at some node $u \in V\cup W$ will reach destination $d$ in less than $2K$ hops with probability at least $1-n^{-C}$.
	\end{restatable}

	The other important ingredient in the proof of \cref{thm:bipartite} is the following technical statement about Markov chains. A proof for this statement is given in \cref{sec:details-bipartite} on page \pageref{proof:markov}. It exploits that, in expectation, the chain drifts towards $0$ with every two further elements.
	
	\begin{restatable}[Markov Chain Aggregation]{theorem}{markovtheorem}
        \label{lem:markov} \label{thm:markovchain}
        Let $\{X_i\}_{i\geq 0}$ be a Markov chain over state space $\mathbb{N}_0$ and $\phi, \psi >0$ be constants with $\phi \cdot \psi < 1$. Let the following be fulfilled for every $i > 0$:
	\begin{enumerate}
		\item $X_{i}$ can be modeled by a sum of Poisson trials that only depends on $X_{i-1}$ \label{drift:item1}
		\item $\Ex{X_{2i + 1}} \leq X_{2i} \cdot \phi $ \label{drift:item2}
		\item $\Ex{X_{2i}} \leq X_{2i-1} \cdot \psi $ \label{drift:item3}
	\end{enumerate}
	Then, there exists a constant $C_{\phi \psi}>1$, such that for any fixed $r > C_{\phi \psi}$ it holds that $\sum_{i=0}^{r} X_i = \BigO{\log(r) \cdot r}$ with probability at least $1-2\exp(-3r)$ as long as $X_0 = \BigO{r}$.
    \end{restatable}
	
	\begin{proof}[Proof of \cref{thm:bipartite}]
		We let $V_G := \{ v ~|~ v \in V \land (v,d) \not \in \mathcal{F}\}$ denote the set of $\emph{good}$ nodes in $V$ that may forward incoming packets directly to destination $d$. Note that each  flow that eventually reaches $d$, does so over some $v \in V_G$. Therefore, in case no flow traverses a cycle, it follows that the node with maximum load will be some $v \in V_G$. In the following we will consider one such fixed node $r \in V_G$. W.l.o.g. we assume that this node lies in $V(K-1)$ such that we can avoid modulo operations. We define $L_j$ to be the set of nodes whose load $r$ receives within exactly $j$ hops. Clearly $L_0 = \{r\}$, and according to the definition of $\mathcal{P}_{B}$ it must hold that $L_1 \subseteq W(K -2)$, $L_2 \subseteq V(K -2)$, $L_3 \subseteq W(K -3)$ and so forth.  Our goal is to bound the values $|L_j|$  which allows us to determine the load that $r$ receives. To that end, we initially assume that the routing entries $\alpha(v)$ of any node $v$ have not yet been uncovered. Observe that, in order to determine, for example, $L_1$ it suffices to uncover the entries of nodes in $W(K -2)$ and check which nodes $w \in W(K -2)$ have $\alpha(w)  = r$. To determine $L_2$, we then uncover entries in $V(K -2)$ and check the number of nodes $v$ in this partition having $\alpha(v) \in L_1$.  A repetition of this approach step-by-step yields the following intermediate result, which we show in \cref{sec:bipartite} on page \pageref{proof:bipartite-markov}.
		\begin{restatable}{observation}{bipartitemarkov}
		\label{obs:bip-martingale}
			The sequence $\{|L_i|\}_{i = 0}^{2K}$ forms a Markov chain with $|L_0| = 1$. Additionally, for $i> 0$ it holds that 
			\begin{enumerate}
				\item $|L_i|$ can be modeled by a sum of Poisson trials depending only on $|L_{i-1}|$
				\item $\Ex{|L_{2i + 1}|} \leq  (3/2) \cdot  |L_{2i}|$ \label{obs:item-bip-odd}
				\item $\Ex{|L_{2i}|} \leq  (1/2) \cdot |L_{2i - 1}|$  \label{obs:item-bip-even}
			\end{enumerate}
		\end{restatable}
		
    \vspace{0.2cm}
	Next we make use of \cref{lem:bip-noloops}. Its statement, together with a union bound application, implies that $\emph{no}$ packet originating from any node travels more than $2K$ hops with probability at least $\geq 1 - |V \cup W| n^{-C} \geq 1-2n^{-(C-1)}$. This implies $|L_{i}| = 0$ for $i \geq 2K$. Hence, our fixed node $r \in V_G$ receives in total $\sum_{i=0}^{2K} |L_i|$ load \whp.  As the sequence $\{ |L_i|\}_{i \geq 0}^{2K}$ is a martingale that follows the properties described in \cref{obs:bip-martingale}, we may apply the Markov chain result \cref{lem:markov} for $r = 2K$. It implies that  $\sum_{i=0}^{2K} |L_i|  = \BigO{\log n \cdot \log \log n}$ with probability at least $1- 2n^{-6C}$. Hence, a union bound application yields that for $\emph{any}$ node $r \in V_G$, we have with probability at least $1- n \cdot (2n^{-(C-1)} - 2n^{-6C}) > 1 - 3n^{-(C-1)}$ that $\mathcal{L}(r) = \BigO{\log n \cdot \log \log n}$.
	As \emph{no} packet starting at any node $V \cup W$ travels in a cycle, the node with maximum load (excluding $d$) must be some node $r \in V_G$ and  \cref{thm:bipartite} follows.
	\end{proof}

\subsection{Lower Bound for the Bipartite Graph}
\label{sec:bipartite-lower-bound}
In the following, we present a different lower bound variant. It also holds in settings where nodes in $W$ do not contribute one initial flow in the all-to-one routing process. However, it only guarantees high load in expectation as opposed to the high probability guarantee of \cref{thm:lower-bound-main}. We will make use of this version in the analysis of the Clos topology. The proof is given in \cref{sec:details-bipartite} on page \pageref{proof:bipartite-lower-bound}.
	
\begin{restatable}{lemma}{bipartitelowerbound}
\label{lem:bipartite-lower-bound}
	Let $G = (V \cup W, E)$ be a complete bipartite graph with $|V| = |W| = n$ and assume that the nodes in $V$ each initiate one flow towards some node $d \in W$. Then, for any local destination-based failover protocol $\mathcal{P}$, there exists a set of failures $\mathcal{F}$ of size $ | \mathcal{F}| \leq \varepsilon \cdot n/\log n$, $\varepsilon > 0$ arbitrary constant, such that, in expectation, the number of nodes with load $\Omega(\log n / \log \log n)$ is at least one.
\end{restatable}

\section{Efficient Protocol for the Clos Topology}
\label{sec:fat-tree}

\newcommand{\bl}[1]{\ensuremath{\text{B}\left(#1\right)}}
\newcommand{\cl}[2]{\ensuremath{\text{C}\left(#1,#2\right)}}
\newcommand{\vcl}[2]{\ensuremath{\text{VC}\left(#1,#2\right)}}

\newcommand{\cli}[3]{\ensuremath{\text{C}\left(#1,#2,#3\right)}}
\newcommand{\vcli}[3]{\ensuremath{\text{VC}\left(#1,#2,#3\right)}}

\newcommand{\tree}[1]{\ensuremath{\text{T}\left(#1\right)}}

\newcommand{\prefix}[2]{\ensuremath{#1|#2}}
\newcommand{\sset}{\ensuremath{\mathbb{S}}_L}

\newcommand{\emptyseq}{\ensuremath{\emptyset}}

\subsection{Topology Description}
\label{sec:clos-definition}

The Clos topology we consider comes with two parameters $k$, the degree of each node in the network, and $L+1$, $L \geq 1$, the number of levels in the network (cf.~also~\cref{fig:fattree-example}). It is constructed as follows.
On level $0$, there are $(k/2)^L$ many nodes and each level $\ell$, $1 \leq \ell \leq L$, consists of $2 (k/2)^L$ many nodes. We assume the nodes in each level to be numbered, starting with $1$. All nodes are then partitioned into \emph{blocks}. We denote such a block by $\bl{S}$, where $S$ is a sequence from the set $\sset$. This set contains all sequences $S = (s_1, s_2, ... ,s_\ell)$ of length $0 \leq \ell \leq L$, where the $s_i$ are integers subject to $s_1 \in [1,k]$ and $s_i \in [1,k/2]$, $i > 1$. The nodes in level $\ell$ are contained in blocks $\bl{S}$ with $S \in \sset \land |S| = \ell$. Each such block contains $(k/2)^{L - \ell}$ many consecutive nodes of level $\ell$. In level $0$ there is only a single block. In case $\ell > 1$ and $S = (s_1, s_2, ... ,s_\ell)$ the block $\bl{S}$ contains the nodes $[o + 1, o + (k/2)^{L-\ell}]$, where $o = (s_1 - 1) \cdot (k/2)^{L-1} + (s_2 - 1) \cdot (k/2)^{L-2} + \dots + (s_{\ell}-1) \cdot (k/2)^{L-\ell}$. 

In the following, we will denote the concatenation operator by $\circ$ and call the blocks $\bl{S \circ i}$, $i \in [1,k/2]$, \emph{children} of $\bl{S}$ (the block in level $\ell = 0$ has $k$ children) and vice-versa $\bl{S}$ the \emph{parent} of the blocks $\bl{S \circ i}$. Edges are only drawn between blocks that have a parent-child relationship. This can be seen in  \cref{fig:fattree-example}, where the blocks are visualized as blue boxes (the block at the top is $\bl{\emptyset}$). We then denote by $\tree{S}$ the subgraph containing all blocks $\bl{S'}$ such that $S$ is a prefix of $S'$. For such a fat-tree $\tree{S}$, we say that it is \emph{rooted} in $\bl{S}$. Note, when compressing each block to a single node and drawing an edge for each parent-child relationship, then the resulting graph becomes a tree such that $\bl{S'''}$ is a successor of $\bl{S''}$ iff $S'''$ is a prefix of $S''$.
In order to describe how edges are drawn, we also define \emph{clusters} such that every block $\bl{S}$ is partitioned into clusters. Each cluster $\cl{S}{i}$, $i\geq 1$, contains the first $i \cdot (k/2)$ consecutive nodes of $\bl{S}$.
Edges are inserted by constructing complete bipartite subgraphs. For all clusters $\cl{S}{i}$, we draw edges from every node in the cluster to the $i$-th node in each of the children of $\bl{S}$ (and vice-versa). We call this set of nodes in the children \emph{vertical cluster} $\vcl{S}{i}$. 
In \cref{fig:fixed-level} we illustrate how these edges are drawn.

\begin{figure}
    \centering
    \scalebox{0.8}{
        \input{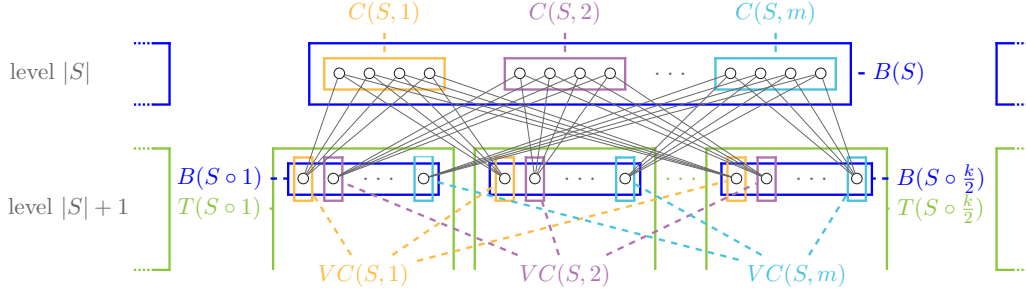}
    }
    \caption{Links between a fixed block $\bl{S}$ and its children. Here $m$ denotes the number of clusters in block $\bl{S}$.} \label{fig:fixed-level}
\end{figure}

Note, from the point-of-view of a fixed node $v$ in some level $\ell$, it resides in exactly one cluster of some block $\bl{S}$ with $|S| = \ell$. Furthermore, in case of $\ell > 0$, it also lies in exactly one vertical cluster, the cluster $\vcl{S'}{i'}$ where $\bl{S'}$ is the parent of $\bl{S}$.

\subsection{Routing Protocol} 

In the following section, we describe how the interval routing protocol of \cref{sec:bipartite} can be adapted to the Clos topology. Note that we only consider topologies with a constant amount of layers, i.e., $L = \Theta(1) > 1$.
To enable interval routing, we employ an additional layer of granularity. That is, we partition each cluster and vertical cluster into $K := (4+L) \log k$ consecutive intervals of size  $I := k/( (8+2L) \log k) = \Theta(k / \log k)$. We denote the $j$-th such interval, $j \geq 0$, of each cluster $\cl{S}{i}$, and the vertical cluster $\vcl{S}{i}$ by $\cli{S}{i}{j}$ or $\vcli{S}{i}{j}$, respectively.  A slight exception to this occurs at level $0$ which only consists of a single block $\bl{\emptyset}$. Here $\emptyseq$ is used to denote the sequence of length $0$. As $\bl{\emptyseq}$ has $k$ children instead of $k/2$, there are $k$ nodes in each cluster $\vcl{\emptyseq}{i}$. These vertical clusters are also split into $K$ many intervals, each containing $2 \cdot I$ nodes.

We focus on all-to-one routing towards some destination $d$ which resides on level $L$ (servers are typically located at the bottom of the Clos topology~\cite{al2008scalable}).  The primary tool to route packets towards $d$ is the sequence $S_d$, which we define as the sequence $S \in \sset$ with length $|S| = L$ that fulfills $\bl{S} = \{d\}$. Note that for each node on level $L$ such a sequence must exist (in \cref{fig:fattree-example} this is visualized as in the bottom layer each node is contained in its own block). We denote by $\prefix{S_d}{i}$, $0\leq i \leq L$ the length $i$ prefix sequence of $S_{d}$. Furthermore, we let $d_{i,j}$ denote the $j$-th node in the block $\bl{\prefix{S_d}{i}}$.
The routing protocol follows the definitions of a local failover protocol given in \cref{sec:local-failover} and equips each node with fitting distributions.

\newcommand{\ruleone}{\hyperref[def:interval-routing]{R1}\xspace}
\newcommand{\ruletwo}{\hyperref[def:interval-routing]{R2}\xspace}

\begin{definition}[Clos Interval Routing]
\label{def:interval-routing}
    Let $v$ be a node $v \in \cli{S}{i}{j}$. Protocol $\mathcal{P}_C$ equips $v$ with the following  distributions to enable routing towards $d$.
    
    \begin{enumerate}
        \item[(R1)] $S$ is \emph{not} a prefix of $S_d$. Let $\bl{S^P}$ denote the parent of $v$'s block $\bl{S}$. Then $v \in \vcli{S^P}{i'}{j'}$ for some $i', j'\geq 1$ and 
        \[
            \Dist{v}{\mathcal{F}_v} {d} = \Uniform \Big( \cli{S^P}{i'}{(j'+1) \mod K} \setminus \mathcal{F}_v  \Big)
        \]
        \item[(R2)] $S$ is a length-$s$ prefix of $S_d$. If $d_{s+1,i} \not \in \mathcal{F}_v$, set $\Dist{v}{\mathcal{F}_v}{d} = \Uniform(\{ d_{s+1,i} \})$. Else, set 
        \[
            \Dist{v}{\mathcal{F}_v}{d} = \Uniform \Big( \vcli{S}{i}{j} \setminus \mathcal{F}_v \Big)
        \]
    \end{enumerate}
\end{definition}

The basic idea behind the routing protocol is to send a packets with destination $d$ from child to parent blocks until they reaches a block $\bl{S}$ such that $S$ is prefix of $S_d$ (\ruleone). Assume now that, after reaching this block $\bl{S}$, the packet lies on a node $v_1$ in interval $\cli{S}{i}{j}$ and $S$ is a length $|S| = s$ prefix of $S_d$. As $v_1 \in \cl{S}{i}$, it is connected to $d_{{s+1},i}$, which lies in $\vcl{S}{i}$. After forwarding the packet to this node, it would then reside on a node in $\bl{\prefix{S_d}{s+1}}$. Note that this block's sequence matches the destination for one more element. However, in case $d_{s+1,i}$ cannot be reached, the only link from $v$ into $\bl{\prefix{S_d}{s+1}}$ is unreachable. In such a case, it is forwarded to some $w_1 \in \vcli{S}{i}{j}$ instead (\ruletwo). As $w_1$ lies on a block $\bl{S'}$, which is a child of $\bl{S}$ that has some sequence $S' \neq \prefix{S_d}{s+1}$, the packet is forwarded according to \ruleone in the next step. Afterwards, it will again lie on a node $v_2$ in $\cl{S}{i}$. However, this time in the interval $\cli{S}{i}{j+1}$. In the next step, the packet is again attempted to be forwarded to $d_{s+1,i}$. Otherwise it is forwarded to $\vcli{S'}{i}{j+2}$ and the procedure repeats. Intuitively, the packet "ping-pongs" between layers $|S|$ and $|S|+1$ until it manages to reach $d_{s+1,i}$, similar as in the protocol for the complete bipartite graph of \cref{sec:bipartite}. As the forwarding partners are chosen u.a.r., it is unlikely for the packet to hit a node with failed link to $d_{s+1,i}$ in each of $\Omega(\log k)$ alternations, and it will eventually hit $d_{s+1,i}$. A visualization of this idea is given in \cref{fig:ping-pong}. We also invite the reader to familiarize her- or himself with the more detailed example we prepared in \cref{sec:clos-example} starting on page \pageref{sec:clos-example}.

\begin{figure}
    \centering
    \scalebox{0.8}{
        \tikzset{
    rect/.style n args={4}{
            draw=none,
            very thick,
            rectangle,
            minimum width=2cm,
            minimum height=0.5cm,
            append after command={
                \pgfextra{%
                    \pgfkeysgetvalue{/pgf/outer xsep}{\oxsep}
                    \pgfkeysgetvalue{/pgf/outer ysep}{\oysep}
                    \def\arg@one{#1}
                    \def\arg@two{#2}
                    \def\arg@three{#3}
                    \def\arg@four{#4}
                    \begin{pgfinterruptpath}
                        \ifx\\#1\\\else
                            \draw[draw,#1] ([xshift=-\oxsep]\tikzlastnode.south east) edge ([xshift=-\oxsep]\tikzlastnode.north east); 
                        \fi\ifx\\#2\\\else
                            \draw[draw,#2] ([yshift=-\oysep]\tikzlastnode.north east) edge ([xshift=0\ifx\arg@three\@empty+\pgflinewidth\fi,yshift=-\oysep]\tikzlastnode.north west); 
                        \fi\ifx\\#3\\\else
                            \draw[draw,#3] ([xshift=\oxsep]\tikzlastnode.north west) edge ([xshift=\oxsep]\tikzlastnode.south west); 
                        \fi\ifx\\#4\\\else
                            \draw[draw,#4] (\tikzlastnode.south east) edge (\tikzlastnode.south west); 
                        \fi
                    \end{pgfinterruptpath}
                }
            }
        },
    txt/.style={
        draw=none,
        rectangle
    },
    switch/.style={
        draw=black!55,
        circle,
        thick,
        inner sep=0pt,
        minimum size=0.175cm
    }
}
\begin{tikzpicture}[very thick]
    \def\RowWidth{0.75}
    \def\ColWidth{2}

    \def\CA{0*\ColWidth}
    \def\CB{1*\ColWidth}
    \def\CC{2*\ColWidth}
    \def\CD{3*\ColWidth}
    \def\CE{4*\ColWidth}
    \def\CF{5*\ColWidth}
    \def\CG{6*\ColWidth}
    \def\CH{7*\ColWidth}
    \def\CI{8*\ColWidth}
   
    \def\RA{0*\RowWidth}
    \def\RB{-1*\RowWidth}
    \def\RC{-2*\RowWidth}
    \def\RD{-3*\RowWidth}
    \def\RE{-4*\RowWidth}
   
    \node [txt, text=black] at (\CA-\ColWidth/10,\RB) {level $|S|\;\;\;\;\;\;$};
   
    \node [txt, text=blue] at (\CC, \RA-\RowWidth/5) {\small$C(S,i,j)$};
    \node [txt, text=blue] at (\CD, \RA-\RowWidth/5) {\small$C(S,i,j+1)$};
    \node [txt, text=blue] at (\CE, \RA-\RowWidth/5) {\small$C(S,i,j+2)$};
    \node [txt, text=blue] at (\CI, \RB)             {$C(S,i)\;\;$};
   
    \node [rect={SkyBlue}   {blue} {blue} {blue}] at (\CB, \RB) {};
    \node [rect={SkyBlue}   {blue} {}          {blue}] at (\CC, \RB) {};
    \node [rect={SkyBlue}   {blue} {}          {blue}] at (\CD, \RB) {};
    \node [rect={SkyBlue}   {blue} {}          {blue}] at (\CE, \RB) {};
    \node [rect={SkyBlue}   {blue} {}          {blue}] at (\CF, \RB) {};
    \node [rect={SkyBlue}   {blue} {}          {blue}] at (\CG, \RB) {};
    \node [rect={blue} {blue} {}          {blue}] at (\CH, \RB) {};
   
    \node [switch, draw=MidnightBlue!70!black!100] at (\CC, \RB) (v1) {};
    \node [switch, draw=MidnightBlue!70!black!100] at (\CD, \RB) (v2) {};
    \node [switch, draw=black!60] at (\CE, \RB) (v3) {};
   
    \node [txt, text=MidnightBlue!70!black!100] at (\CC+\ColWidth/6, \RB+\RowWidth/10) {$v_1$};
    \node [txt, text=MidnightBlue!70!black!100] at (\CD+\ColWidth/6, \RB+\RowWidth/10) {$v_2$};
    \node [txt, text=MidnightBlue!70!black!100] at (\CE+\ColWidth/6, \RB+\RowWidth/10) {$v_3$};
   
   
    \node [txt, text=black] at (\CA-\ColWidth/10, \RD) {level $|S| + 1$};
   
    \node [txt, text=blue] at (\CC, \RE+\RowWidth/5) {\small $VC(S,i,j)$};
    \node [txt, text=blue] at (\CD, \RE+\RowWidth/5) {\small $VC(S,i,j+1)$};
    \node [txt, text=blue] at (\CI, \RD)             {\small $VC(S,i)$};

    \node [txt, text=black!60]  at (\CG, \RE+\RowWidth/5) {$d_{|S|+1,i}$};
   
    \node [rect={SkyBlue}   {blue} {blue} {blue}] at (\CB, \RD) {};
    \node [rect={SkyBlue}   {blue} {}          {blue}] at (\CC, \RD) {};
    \node [rect={SkyBlue}   {blue} {}          {blue}] at (\CD, \RD) {};
    \node [rect={SkyBlue}   {blue} {}          {blue}] at (\CE, \RD) {};
    \node [rect={SkyBlue}   {blue} {}          {blue}] at (\CF, \RD) {};
    \node [rect={SkyBlue}   {blue} {}          {blue}] at (\CG, \RD) {};
    \node [rect={blue} {blue} {}          {blue}] at (\CH, \RD) {};
   
    \node [switch, draw=MidnightBlue!70!black!100] at (\CC, \RD) (w1) {};
    \node [switch, draw=MidnightBlue!70!black!100] at (\CD, \RD) (w2) {};
    \node [switch, draw=MidnightBlue!70!black!100] at (\CG, \RD) (w3) {};
   
    \node [txt, text=MidnightBlue!70!black!100] at (\CC+\ColWidth/6, \RD-\RowWidth/10) {$w_1$};
    \node [txt, text=MidnightBlue!70!black!100] at (\CD+\ColWidth/6, \RD-\RowWidth/10) {$w_2$};

    \draw[->, dashed, thick, draw=purple] (v1) -- (w3);
    \draw[->, dashed, thick, draw=purple] (v2) -- (w3);
   
    \draw[->, thick, draw=MidnightBlue!70!black!100] (v1) -- (w1);
    \draw[->, thick, draw=MidnightBlue!70!black!100] (w1) -- (v2);
    \draw[->, thick, draw=MidnightBlue!70!black!100] (v2) -- (w2);
    \draw[->, thick, draw=MidnightBlue!70!black!100] (w2) -- (v3);
    \draw[->, thick, draw=purple]   (v3) -- (w3);
   
    \draw[<->, thick, draw=black!60] (\CH+\ColWidth*3/8, \RB - \RowWidth/2.3) -- (\CH+\ColWidth*3/8, \RD + \RowWidth/2.3);
    \node [txt, text=black!60] at (\CI-\ColWidth/6, \RC) {Bipartite};

\end{tikzpicture}
    }
    \caption{"Ping-Pong" of packet starting at $v_1$ in a block $\bl{S}$ where $S$ is prefix of $S_d$} \label{fig:ping-pong}
\end{figure}
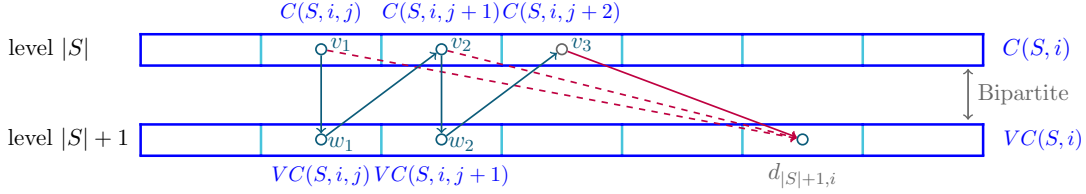

\begin{theorem}
\label{thm:a-tree-main}
		Let $G = (V,E)$ be a Clos topology with degree $k$ and $L+1$ levels for some constant $L > 1$. Consider the routing protocol $\mathcal{P}_C$ and assume all-to-one routing towards some destination $d$ on level $L$. 
		Assume the adversary chooses its set of failures $\mathcal{F}$ such that the following holds for every triple $(S,i,j)$ where $S \in \sset$ with $0 \leq |S| \leq L-1$, $1 \leq i \leq L - |S|$ and $0 \leq j < K$:
		\begin{enumerate}
		    \item  $\forall w \in \vcl{S}{i}:$  $| \{ v \in \cli{S}{i}{j} ~|~ w \in \mathcal{F}_v\}| \leq I/3$
		    \item  $\forall v \in \cl{S}{i}:$   $| \{ w \in \vcli{S}{i}{j} ~|~ v \in \mathcal{F}_w\}| \leq I/3$
		\end{enumerate}
Then, with probability $1-\BigO{k^{-4}}$, every node $u \in V \setminus \{ d\}$ has $\mathcal{L}(u)  = \BigO{k^{L-1} \cdot \log n \log \log n}$, even if the adversary knows $\mathcal{P}_C$ and $d$.
\end{theorem}
While the requirement on the failure set $\mathcal{F}$ may seem restrictive at first, it simply states that in every interval at most $I/3 = \Theta(k / \log k)$ many nodes may have failed edges to the same node. Note that $I/3$ failures from nodes of the same interval are simultaneously allowed to many different nodes.

\subsection{\texorpdfstring{Analysis of \cref{thm:a-tree-main}}{Analysis}}
\label{sec:clos-analysis}

Throughout this proof, we consider the destination $d$ as well as the set of failed edges placed by the adversary $\mathcal{F}$ to be fixed (we assume this set to adhere to the requirements of \cref{thm:a-tree-main}). As described in \cref{sec:local-failover}, each node $v$ draws its routing entry $\alpha(v,\mathcal{F}_v, d)$ from $\Dist{v}{\mathcal{F}_v}{d}$ which is specified in  \cref{def:interval-routing}. As we consider the set of failures $\mathcal{F}$ as well as $d$ to be fixed, we use the abbreviation $\alpha(v) := \alpha(v, \mathcal{F}_v,d)$.

\paragraph{Staggered Load Calculation:}
 In the following, we will not immediately uncover all entries $\alpha(v)$ required to determine the load $\mathcal{L}(v)$ some node receives. Instead, we will uncover these entries step-by-step. To that end, we extend our notion of \emph{load} defined in \cref{sec:local-failover}, to also apply in cases where some entries are still left covered. Flows that arrive at a node $v$ with a still covered entry $\alpha(v)$, are assumed to be \emph{stopped} and only contribute to load of nodes that lie on the path the flows takes to reach $v$. As soon as the entry of $v$ is uncovered, all stopped flows continue to flow until they either reach $d$, or hit another node with a covered entry. It is easy to see that by increasing the number of uncovered entries, the load at any node can only increase, and, after uncovering all entries of nodes $v \neq d$, we end up with the  notion of load defined in \cref{sec:local-failover}. This staggered uncovering of entries allows us to develop a bound on the load step-by-step and helps us circumvent dependencies of the traffic flow in different parts of the topology.

\begin{restatable}{lemma}{atreeinduction}
\label{lem:a-tree-induction}
    Let $\ell$ be an integer in $[0,L-1]$. Then, after uncovering all entries besides those of nodes in $\tree{\prefix{S_d}{\ell}}$, the following holds with probability $\geq 1 - 4\ell \cdot k^{-4}$:
    \begin{enumerate}
        \item no flow travels in a cycle, i.e., $\mathcal{L}(v) < \infty$ for all nodes $v$ \label{item:a-tree-1}
        \item flows of nodes with uncovered entries are stopped at some node $v \in \bl{\prefix{S_d}{\ell}}$ \label{item:a-tree-2}
        \item all nodes, including those in $\bl{\prefix{S_d}{\ell}}$, have load $\BigO{k^{\ell}}$ \label{item:a-tree-3}
    \end{enumerate}
\end{restatable}
\begin{proof}[Sketch of Proof]
    The proof uses induction over the levels $\ell = 0$ to $L-1$ in the following way. In each step of the induction we uncover the edges in $\tree{\prefix{S_d}{\ell}} \setminus \tree{\prefix{S_d}{\ell+1}}$. 
    In the induction hypothesis, we assume that the lemma holds up to some 
    $\ell \in [0,L-2]$ and in the induction step, we show that the statement also holds for $\ell+1$. The base case (i.e., before we uncover any entries) trivially holds. 
    
    In order to perform the induction step, we use a two-step approach. First, we uncover the edges in $\mathcal{T}_\ell = \{ \tree{\prefix{S_d}{\ell} \circ i} ~|~ 1 \leq i \leq k/2 \land \prefix{S_d}{\ell} \circ i \neq \prefix{S_d}{\ell+1} \}$ (note that if $\ell = 0$, then $i$ is in the range $1, \dots , k$). As an example, if $\ell=0$ and $d$ is the last vertex in level $3$ in \cref{fig:fattree-example}, then the set $\mathcal{T}_\ell$ contains the subtrees rooted in the first three blocks of level $1$. If $\ell=1$ ($d$ remains the same node in \cref{fig:fattree-example}), then $\mathcal{T}_\ell$ is the subtree rooted in the $7$th block of level $2$. After uncovering the edges in $\mathcal{T}_\ell$, we show that every vertical cluster in the blocks on level $\ell+1$ in $\mathcal{T}_\ell$ contain $\BigO{k}$ load.
    
    In the second step, we uncover the edges between levels $\ell$ and $\ell+1$ in $\tree{\prefix{S_d}{\ell}}$. As a result, we obtain that the statement holds for $\ell+1$. This second step heavily uses the properties of the failover routing algorithm in complete bipartite graphs. The full proof is given in \cref{sec:details-fat-tree} on page \pageref{proof:a-tree-induction}.
\end{proof}

\begin{proof}[Proof of \cref{thm:a-tree-main}]
    Let $S_{L-1} := \prefix{S_d}{L-1}$. We start with an application of \cref{lem:a-tree-induction} for $\ell = L-1$. When uncovering all entries except those in $\tree{S_{L-1}}$, this implies that the flows of all nodes outside $\tree{S_{L-1}}$ enter $\tree{S_{L-1}}$ at its root $\bl{S_{L-1}}$ without causing load higher than $O(k^{L-1})$ \whp. 
    Note that the root $\bl{S_{L-1}}$ of $\tree{S_{L-1}}$ consists of $k/2$ nodes and therefore only a single cluster. More precisely the following holds.

    \begin{observation}
    \label{obs:last-tree}
    $\tree{S_{L-1}}$ is a complete bipartite graph consisting of the clusters $\cl{S_{L-1}}{1}$ and $\vcl{S_{L-1}}{1}$. These clusters each have size $k/2$ and $d \in \vcl{S_{L-1}}{1}$.
    \end{observation}

    This enables us to apply results from the bipartite graph section (\cref{sec:bipartite}). As $S_{L_1}$ is a prefix of $S_d$, all flows starting from nodes in $\cl{S_{L-1}}{1} \cup \vcl{S_{L-1}}{1}$ will "ping-pong" between the clusters until $d$ is reached (see \cref{fig:ping-pong}).
    Note that the path taken by the packets in $\tree{S_{L-1}}$ according to $\mathcal{P}_C$ is exactly the same as if  the nodes in $\tree{S_{L-1}}$ would follow the bipartite routing protocol $\mathcal{P}_B$ instead (described in \cref{sec:bipartite} with $C = (4+ L)$). The main result of that section, \cref{thm:bipartite}, then implies that at most $\BigO{\log k \cdot \log \log k}$ load is created \whp. However that result assumes that each node in $\tree{S_{L-1}}$ starts with only $1$ flow, while in our case up to $O(k^{L-1})$ flows start from a single node in $\cl{S}{1}$ as soon as the entries in $\tree{S_{L-1}}$ are uncovered. Thus, we obtain a maximum load of $O(k^{L-1} \cdot \log k \log \log k)$.
\end{proof}

\subsection{Lower Bound for the Clos Topology}

Managing load under the all-to-one traffic pattern is inherently challenging even in highly-connected Clos topologies. To illustrate this, we construct a simple congestion lower bound of $\Omega(k^{L-1})$, which holds even in the absence of link failures and does not rely on our notion of local failover routing. Assume that all-to-one routing towards a node $d$ in level $L$ is performed. This destination node $d$ is incident to only $k/2$ many nodes, all of which lie in level $L-1$. All flows need to travel over one of these $k/2$ nodes to reach $d$. As the Clos topology contains $\Omega(k^{L})$ many nodes in total and each node sends a flow towards $d$, it follows that one of the $k/2$ many neighbors of $d$ must accumulate $\Omega(k^{L-1})$ flows.

\paragraph{Refined Lower Bound}
In the remainder of this section, we present an improved lower bound that is targeted towards a class of  routing protocols that exhibit the following properties.

\begin{definition}[Fairly Balanced and Shortest Path Routing]
\label{def:fairly-balanced}
    Let $\mathcal{P}$ be a local failover protocol operating in the Clos topology, assume $\mathcal{F} = \emptyset$, i.e., no edges are failed, and assume that all-to-one routing towards any arbitrary destination $d$ on level $L$ is performed. We call $\mathcal{P}$ a \emph{fairly balanced} protocol if, \whp, it holds that $\mathcal{L}(v) = \Theta(k^\ell)$ for  $v \in \bl{\prefix{S_d}{\ell}}, 1 \leq \ell \leq L$, and $\mathcal{L}(v) = \BigO{\polylog k}$ otherwise.
    
    \noindent Furthermore, we say $\mathcal{P}$ is a \emph{shortest path} protocol, or $\mathcal{P}$ \emph{forwards over shortest paths} if the following holds for every node $v \in V \setminus \{d\}$. Let $\mathcal{F}$ be a set of failures (which might be empty) such that $|\mathcal{F}_v| \leq I / 3 = \BigO{k / \log k}$. Then, the routing entry $\alpha(v, \mathcal{F}_v, d)$ must always lie on a shortest path to $d$ in the graph $(V,E \setminus \mathcal{F}_v)$. 
\end{definition}

While assuming these properties limits the generality of the following lower bound, they are natural and realized by the standard equal-cost multipath  protocol ECMP~\cite{singh2015jupiter,kabbani2014flowbender} which also underlies the widely-used routing protocols OSPF.\footnote{Specifically, ECMP balances flows across shortest paths by default, and upon a failure, locally re-hashes to redistribute flows across the remaining shortest paths to the destination.} Intuitively, fairly balanced means that the load that has to enter some block at a particular level of the Clos topology is "fairly" balanced among the nodes of the particular block. That is, each such node receives the same load up to constant factors. Note that protocols, which do not exhibit the fairly balanced property, seem unnatural as they may generate load situations in which some nodes are heavily affected by flows while others (on the shortest path from a level $0$-node to the destination) remain idle. 

However, note that ECMP protocol may generate cycles with probability $1/\polylog n$ if the number of failures is $\Omega(k/\log k)$. 
For protocols that exhibit above properties, one can construct a load lower bound of $\Omega(k^{L-1} \log n/\log \log n)$. To achieve this bound, the adversary only fails edges in $\tree{\prefix{S_d}{L-1}}$, which is a complete bipartite graph (see \cref{obs:last-tree}) consisting of $k$ nodes partitioned into $\cl{\prefix{S_d}{L-1}}{1}$ and $\vcl{\prefix{S_d}{L-1}}{1}$. As we only consider fairly balanced protocols it follows that each node in $\cl{\prefix{S_d}{L-1}}{1}$ (the one partition of the complete bipartite graph) receives a load of $\Theta(k^{L-1})$. 
Due to shortest path routing, no load will ever leave this bipartite graph again. Then, a set of edges incident to the destination can be failed such that a load of $\Omega(k^{L-1} \log k/\log \log k)$ is generated. The corresponding results can be found in \cref{lem:bipartite-lower-bound} of \cref{sec:bipartite-lower-bound}.
We then show in \cref{obs:a-tree-fairly-stable} that our routing protocol $\mathcal{P}_C$ indeed fulfills the properties in \cref{def:fairly-balanced}. This implies that the result in \cref{thm:a-tree-main} is tight up to a $\polyloglog$ factor. The proofs are in \cref{sec:details-fat-tree}. 

\begin{restatable}{lemma}{atreelowerbound}
\label{lem:a-tree-lower-bound}
    Let  $\mathcal{P}$ by a fairly balanced protocol that operates in a Clos topology with $L = \Theta(1)>1$ layers and only forwards over shortest paths to some destination $d$ on level $L$. Then, there exists a set of failures $\mathcal{F}$ with $|\mathcal{F}| \leq I / 3 = \BigO{k / \log k}$ such that, in expectation, at least one node $v \neq d$ has $\mathcal{L}(v) = \Omega(k^{L-1} \log k / \log \log k)$. 
 \end{restatable}
\begin{restatable}{lemma}{atreefairlybalanced}
\label{obs:a-tree-fairly-stable}
    The protocol $\mathcal{P}_C$ defined in \cref{def:interval-routing} is a fairly balanced shortest path protocol.
\end{restatable}

\section{Technical Details}
\label{sec:details}

\subsection{Lower Bound Analysis}
\label{sec:details-lower-bound}

\lbcaseone*

\begin{proof}\label{proof:lb-case1}
	Let  $\mathcal{C}$ be a cycle in $\Glog$.  We start by showing that some node $v$ on $\mathcal{C}$ has at least $\ell := (1/10)  \log n  / \log \log n$ load with probability $\geq n^{-4/10}$.
	
    \case{1} $\mathcal{C}$ has length $\geq \ell$. Then, $\mathcal{C}$ must contain a path $(v_1, v_2 , ... , v_\ell)$ of length $\ell$. We set  $\mathcal{F}_{v_i} = \{ d\}$ for each $v_i$, $1 \leq i \leq \ell$, in other words, we fail for every node on the path the link directly connected to $d$. By \cref{obs:edge-alpha} it follows that any edge $(v_i, v_{i+1})$, $1 \leq i < \ell$, will be in $\Galpha$ with probability $> 1 / \log^4 n$. Therefore, the whole path will be in $\Galpha$ with probability at least $1/\log^{4 \cdot \ell} n = n^{-4 /10}$. This event implies that $v_\ell$ is reached by at least $\ell$ nodes in $\Galpha$ and by \cref{obs:alpha} this directly yields $\mathcal{L}(v_\ell) \geq \ell$.
	
	\case{2} $\mathcal{C}$ has length $< \ell$. Let $(v_1, v_2 , ... , v_{k-1} , v_k) = \mathcal{C}$ where $v_1 = v_k$. In this case, we set $\mathcal{F}_{v_i} = \{ d\}$ for any $v_i$ with $1 \leq i \leq k$. Again, by \cref{obs:edge-alpha} it follows that any edge $(v_i, v_{i+1})$, $1 \leq i < \ell$ will be in $\Galpha$ with probability at least $1/\log^{4k} n > 1/\log^{4\ell} = n^{-4 / 10}$. Note that such an event implies that the whole cycle $\mathcal{C}$ will be in $\Galpha$. In such a case, each node $v_i$ has $\mathcal{L}(v_i) = \infty$.
	
	As we consider the case of $|\Rootsett| - |\Rootset| \geq \sqrt{n}$ there must be at least $\sqrt{n}$ cycles in $\Glog$,  which have been resolved in step 2) of the generation of $\Glogt$ (see \cref{def:glogt}).
	At the start of this step, no node has degree larger than $1$ in $\Glogt$. It follows that these $\geq \sqrt{n}$ cycles must be node-disjoint. Therefore, each of theses cycles has \emph{independent} probability at least $n^{-4 / 10}$ to generate a node with at least load $\ell$ when spending $ \leq \ell$ edge failures. Considering a fixed subset of size $\sqrt{n}$ of such cycles and spending $\leq \ell$ edge failures per cycle, a Chernoff bound application yields that, \whp, some cycle will have a node with $\mathcal{L}(v) > \ell$.
\end{proof}

\lbcaseonetwo*
\begin{proof} \label{proof:case2-1}
	 We place edge failures such that $\mathcal{F} = \{ (v,d) ~|~ v \in \Rootset'\}$ and distinguish two cases.
	 
	\case{1} $\exists w \in V$ such that $\sum_{v \in \Rootset'} f_v(w) > \log^2 n$. Let $Z$ be the random variable that denotes  the number of edges $(v,w)$ in $\Galpha$ such that $v \in \Rootset'$. This allows us to write $Z = \sum_{v \in \Rootset'} Z_v$ where $Z_v$ is an indicator random variable with $Z_v = 1$ iff $(v,w)$ in $\Ealpha$. By \cref{obs:pdf} we have that the $Z_v$ are independent, and $\Pr[Z_v = 1] = f_v(w)$. Therefore, we may apply Chernoff bounds to bound $Z$ and, by assumption of this case, we have $\Ex{Z} > \log^2 n$. Therefore, $Z > \log^{2} n (1 - o(1))$ \whp. Hence $w$ has out-degree at least $\log^2 n (1-o(1))$ in $\Galpha$, which immediately implies that $\mathcal{L}(w) = \log^2 n (1 - o(1))$ and the second statement of the lemma follows.
	
	\case{2} $\forall w \in V$ it holds $\sum_{v \in \Rootset'} f_v(w) \leq \log^2 n$. In this case we focus on a set of nodes $\mathcal{R}$ instead of a single node.
	Let $\mathcal{R} := \{ w \in V ~|~ \sum_{v \in \Rootset'} f_v(w) > \varepsilon / (2 \log n)\}$. By a counting argument, we have $|\mathcal{R}| > n^{7/8}$. To show this, we assume the contrary. Observe that $\sum_{w \in V} f_v(w) = 1$ for any node $v \in V$, as $f_v$ is a PDF. We use this in the following inequality chain.
	\[
		\frac{\varepsilon n}{\log n} =\sum_{v \in \Rootset'}  \sum_{w \in V}  f_v(w)  = \sum_{w \in V}    \sum_{v \in \Rootset'}  f_v(w) < ( n - n^{7/8}) \cdot \frac{\varepsilon}{2\log n} + n^{7/8} \cdot \log^2 n < \frac{\varepsilon n}{\log n}.
	\]
	To derive the third step, we combined our assumption of $|\mathcal{R}| \leq n^{7/8}$ with the fact that \emph{no} node $w \in V$ has $\sum_{v \in \Rootset'} f_v(w) > \log^2 n$, which is the assumption of this case. Clearly this is a contradiction, which implies that $|\mathcal{R}| >n^{7/8}$ must hold. 
	
	Until further notice, our analysis will now consider a fixed node $w \in \mathcal{R}$.
	We know that $\sum_{v \in \Rootset'} f_v(w) > \varepsilon/(2 \log n)$, however, each value $f_v(w)$ might contribute a different amount to this sum. Another thing we know is that $f_v(w) < 1 / \log^4 n$ as $v \in \Rootset'$ are (reverse) roots in $\Glog$. To allow for a more fine-grained categorization we define for $i \in \mathbb{N}_0$
	\[ 
		\Rootset'(i) = \{ v ~|~ v \in \Rootset' \land f_v(w) \in (1/\log^{4 + i} n , 1/\log^{ 5+ i} n] \}.
	\]
	Observe that $\bigcup_{i \in \mathbb{N}_0} \Rootset'(i) = \Rootset'$. We now show that there is a partition with $|\Rootset(i)| > \log^{i} n$ for some $i < \log n / \log \log n$. This follows from a counting argument, as assuming otherwise yields that 
	\begin{align*}
		\sum_{v \in \Rootset'} f_v (w) &= \sum_{i \in \mathbb{N}_0} \sum_{v \in \Rootset'(i)} f_v(w) \\
		&< \sum_{i < \log n / \log \log n} \left( \frac{1}{\log^{4 + i } n } \cdot \log^{i} n \right)  + \sum_{i \geq \log / \log \log n} \left(  \frac{1}{\log^{4+ i} n }\cdot n \right) \\
		&< \frac{\log n}{\log \log n} \cdot \frac{1}{\log^4 n} + \BigO{\frac{1}{\log^4 n}} < \frac{\varepsilon}{2\log n}.
	\end{align*}
	Note, when bounding the second sum in the second step, we used the trivial bound $|\Rootset'(i)| \leq n$. This sequence of inequalities leads to a contradiction as $w \in \mathcal{R}$ implies $\sum_{v \in \Rootset'} f_v(w) > \varepsilon / (2\log n)$. Therefore, there must be a $j< \log n / \log \log n$ such that $|\Rootset'(j)| > \log^j n$.
Similar as in the first case, we now define for $v \in \Rootset'(j)$ the indicator random variable $Z_v$ where $Z_v = 1$ iff $(v,w) \in \Ealpha$. By \cref{obs:pdf} it follows that the $Z_v$ are independent and $\Pr[Z_v = 1] = f_v(w) > 1/\log^{4+j} n$. Therefore, $Z = \sum_{v \in \Rootset'(j)} Z_v$ is stochastically minorized by $\Bin( \log^{j} n \,,\, 1/\log^{4+j} n)$ and for  $\ell := (1/10) \cdot \log n / \log \log n$, we get
	\begin{align*}
		\Pr \left[ Z \geq \ell  \right] \geq \Pr \left[\Bin(\log^{i} n ,  1/\log^{4+i} n) \geq \ell \right] \geq \Pr \left[ \Bin(\log^i n, 1/ \log^{4+i} n) =  \ell  \right]
	\end{align*}
	\begin{align*}
		> \left(  \frac{\log^i n}{\ell} \right)^{\ell} \cdot \left( \frac{1}{\log^{4 + i}n} \right)^{\ell} \cdot \left(1 - \frac{1}{\log^4 n} \right) 
		= \left( \frac{1}{\ell \log^{4} n } \right)^{\ell} (1 - o(1))>  \frac{1}{n^{6 / 10}}
	\end{align*}
	Hence, one fixed node $w \in \mathcal{R}$ is incident to at least $(1/10) \log n / \log \log n$ nodes $v \in \Rootset'$ with probability at least $n^{-6 / 10} = n^{-6/10}$. As $|\mathcal{R}| > n^{7/8}$, the lemma's first statement follows.
\end{proof}

\lbcasetwotwo*

\begin{proof}\label{proof:case2-2}
	Let $\{v_1, v_2 , ... , v_{\varepsilon n/\log n} \} = \Rootset'$. For $v_i$ we define $X_i$ to be the random variable denoting the node such that $(v,X_i) \in \Galpha$ if $\mathcal{F}_v = \{d\}$.  That is, $X_i$ is the node that $v_i$ forwards packets to in case the direct link to $d$ is failed. Note that the random variables $X_i$ are independent. We define the function $f$, where $f(X_1, X_2 , ... , X_{\varepsilon n/\log n})$ denotes the number of nodes $w \in V$ that fulfill the following condition: there exist at least $(1/10) \log n/ \log \log n$ many edges $(v,w)$ with $v \in \Rootset'$ in $\Galpha$. In the following we call such a node \emph{popular}. In \cref{lem:case-2-1} we established that $\Ex{f} > n^{7/8}$.
	Additionally, observe that $f$ is $1$-Lipschitz. That is, $
	    |f(x_1, ..., x_i, ..., x_{\varepsilon n/\log n}) - f(x_1, ..., x_i', ..., x_{\varepsilon n/\log n})| \leq 1.$
	This is because changing the value of $x_i$ only changes the edge $(v_i, x_i)$ in $\Galpha$ to $(v_i, x_i')$. In the worst case, this can cause at most one node to become popular or un-popular respectively. This enables us to employ the \emph{method of bounded differences} (see \cref{thm:mcdiarmid} -- cf. \cite{MRT12}). This inequality states that, for any $t> 0$, 
	\[
		\Pr \Big[ |f - \Ex{f}| \geq t \Big] \leq 2\exp \left( \frac{2t^2}{\sum_{i=1}^{\varepsilon n/\log n} c_i^2} \right) = 2\exp \left( \frac{2t^2}{ \varepsilon n/ \log n} \right)
	\]
	where $c_i = 1$ for all $i$ follows as $f$ is $1$-lipschitz. Setting $t = n^{7/8}  - 2$ easily yields that, \whp, $f > 1$.
	Hence, \whp, there will be at least one popular node $w$. This node is reached by $(1/10) \cdot  \log n / \log \log n$ nodes in $\Galpha$ and from \cref{obs:load} it follows that $\mathcal{L}(w) \geq (1/10) \cdot \log n/ \log \log n$.
\end{proof}

\lemcasethreeone*

\begin{proof}\label{proof:case3-1}
	Let $\mathcal{S}$ be a subset of size $\sqrt{n}$, containing  $\sqrt{n}$ reverse trees of $\Glogt$ that have size at least $(1/10)  \cdot \log n / \log \log n$. Fix one such reverse tree $\mathcal{T}_R$ contained in $\mathcal{S}$. The basic idea is to fail edges of the form $(v,d)$ in $\mathcal{T}_R$, which causes $\mathcal{T}_R$ to appear in $\Galpha$ with a certain probability. However, in order to avoid using  more failures than necessary, we first cut out a (reverse) subtree of size $(1/10) \cdot  \log n / \log \log n$ from $\mathcal{T}_R$. We call this tree $\mathcal{T}_R'$. We let $r'$ be the reverse root of this subtree and set $\mathcal{F}_v = \{ d\}$ for every node $v \in \mathcal{T}_R'$, in other words for every $v \in \mathcal{T}_R'$ we fail the edge connection $v$ to destination $d$. As $\mathcal{T}_R'$ is a reverse tree, each node has at most one outgoing edge. Therefore, according to \cref{obs:edge-alpha-t}, each edge of $\mathcal{T}_R'$ will appear in $\Galpha$ with \emph{independent} probability $1/\log^4 n$. Hence, the whole subtree $\mathcal{T}_R'$ is contained in $\Galpha$ with probability at least
\[
	\log^{-4 \cdot  (1/10) \log n / \log \log n} n = n^{-4/10} = n^{-2/5},
\]
which implies by \cref{obs:load} that $\mathcal{L}(r') \geq (1/10) \cdot  \log n / \log \log n$. Summarizing, by failing the edges of a subtree of $\mathcal{T}_R$ we create a high load with probability at least $n^{-2/5}$. Recall, in $\mathcal{S}$ we had $\sqrt{n}$ such trees. As these trees are completely node-disjoint, this probability of $n^{-2/5}$ can be established for each of these $\sqrt{n}$ reverse trees \emph{independently}. By a Chernoff bound application it follows that, when failing $(1/10) \cdot \log n / \log \log n$ edges in each of the trees in $\mathcal{S}$, a load of $\mathcal{L}(v) > (1/10) \cdot \log  n / \log \log n$ will be created \whp.
 \end{proof}

\lemcasethreetwo*

\begin{proof} \label{proof:case3-2-details}
First consider the case of some node $w$ having $\indeg(w) \geq \log^5 n$ in $\mathcal{T}_R$. Fix this node $w$. In this case there must be a set of $\log^5 n$  nodes $v$ such that edges of form $(v,w)$ lie in $\Glogt$. We now set $\mathcal{F}_v = \{d\}$ for each of these nodes $v$. By \cref{obs:edge-alpha-t}, each of these $\log^5 n$ edges will appear in $\Galpha$ with independent probability at least $1/ \log^4 n$.  From a Chernoff bound application, it follows  that at least $\log n / 2$ of these edges will indeed appear in $\Glogt$ \whp. Therefore, $w$ is reached by at least $\log n / 2$ nodes, implying that $\mathcal{L}(w) > \log n / 2$ and the proof would be finished.

\medskip 

For the remainder of the proof we therefore assume that there is no node $w$ such that $\indeg(w) > \log^5 n$ in $\mathcal{T}_R$. We continue with the following observation, which enables us to "cut out" reverse subtrees from reverse trees.
\begin{observation}
\label{obs:tree-cut}
	Let $\mathcal{T}$ be a reverse tree consisting of $m > (1/10) \cdot  \log n / \log \log n$ nodes each with $\indeg(v) < \log^5 n$. Then, there exist subgraphs $\mathcal{T}_{cut}$ and $\mathcal{F}_{rest}$ of $\mathcal{T}$ such that the following holds.
	\begin{enumerate}
		\item $\mathcal{T}_{cut}$ is a reverse tree with $(1/10) \cdot  \log n / \log \log n$ nodes.
		\item $\mathcal{F}_{rest}$ is a forest of reverse trees each of size at least $(1/10) \cdot  \log n / \log \log n$ nodes and $\mathcal{F}_{rest}$ contains at least $m - \log^{7} n$ nodes.
		\item $\mathcal{T}_{cut}$ and $\mathcal{F}_{rest}$ do not share any nodes.
	\end{enumerate}
\end{observation}
\begin{proof}
	We start by constructing the set of nodes $V_{cut}$ which will later on induce the reverse tree $\mathcal{T}_{cut}$. Let $r$ be the root of $\mathcal{T}$, and let $L_i$, $i \in \mathbb{N}_0$,  denote the set of nodes that have distance $i$ to $r$ in $\mathcal{T}$. We start by adding $L_0 = \{ r\}$ to $V_{cut}$, followed by nodes of the set $L_1$, then $L_2$ and so forth. We include these sets until $|V_{cut}| = (1/10) \cdot  \log n / \log \log n$. Note, for the last level $L_{i^*}$ that we add to $V_{cut}$, we may only add a subset of the nodes in $L_{i^*}$ to $V_{cut}$ as we want to exactly reach $|V_{cut}| = (1/10) \cdot  \log n / \log \log n$. This subset may be chosen arbitrarily. The order of nodes that we add to $V_{cut}$ can be seen as the order in which a breath-first search would traverse the tree. We now define  $\mathcal{T}_{cut}$ to be the subgraph induced in $\mathcal{T}$ by $V_{cut}$. Due to the construction of $V_{cut}$, it follows that $\mathcal{T}_{cut}$ is a reverse tree.
	
	\medskip

	The construction of $\mathcal{F}_{rest}$ starts with a copy of $\mathcal{T}$. Then, we remove all nodes from $\mathcal{F}_{rest}$ that are in $V_{cut}$ and also all edges involving some node $v \in V_{cut}$. Observe that $\mathcal{F}_{rest}$ is now a forest of reverse trees. Additionally, for each reverse root $r'$ in $\mathcal{F}_{rest}$ there must be an edge $(r',v)$ in $\mathcal{T}$ such that $v \in V_{cut}$. As we assume that $\forall v \in V: \indeg(v) < \log^5 n$ it must also hold that any  $v \in V_{cut}$ has $\indeg(v) < \log^5 n$. Therefore, the forest $\mathcal{F}_{rest}$ has at most $\log^5 n \cdot |V_{cut}| = (1/10) \cdot  \log^6 n / \log \log n$ many roots (this number also corresponds to the amount of reverse trees in $\mathcal{F}_{rest}$). Hence, it is impossible that more than $(1/10)^2 \log^7 n / (\log \log n)^2$ nodes may lie in reverse trees of size $<(1/10) \cdot \log n / \log \log n$. In the second and final step of the construction of $\mathcal{F}_{rest}$, we remove all reverse trees of size $< (1/ 10) \cdot  \log n / \log \log n$ from $\mathcal{F}_{rest}$. As just established, the number of nodes in $\mathcal{F}_{rest}$ after this removal still lies above 
	\[
		m - |V_{cut}| - (1/100) \log^7 n / (\log \log n)^2 > m - \log^7 n.
	\]
The proof of the third statement of the observation is trivial as $\mathcal{F}_{rest}$ is constructed by removing the nodes of $\mathcal{T}_{cut}$.
\end{proof}
The idea of the remaining proof is to repeatedly apply \cref{obs:tree-cut} and create a set $\emph{Trees}$ that contains a set of node-disjoint reverse subtrees of $\mathcal{T}$. 
Suppose that in some step $i$, $0 \leq i < \sqrt{n} / \log^7 n$, we have a forest $\mathcal{F}^{(i)}$ that has at least $\sqrt{n} - i \cdot \log^7$ nodes and only contains reverse trees of size $(1/10) \cdot  \log n / \log \log n$. Then, we can pick an arbitrary reverse tree $\mathcal{T}^{(i)}$ from $\mathcal{F}^{(i)}$. Let $m_i > (1/10) \cdot \log n / \log \log n$ denote the size of this tree. We apply  \cref{obs:tree-cut} to this tree, which yields us $\mathcal{T}_{cut}^{(i)}$ and $\mathcal{F}_{rest}^{(i)}$. We now add $\mathcal{T}_{cut}^{(i)}$ to the set $\emph{Trees}$ . Additionally, we define the forest $\mathcal{F}^{(i+1)} := (\mathcal{F}^{(i)} \setminus \mathcal{T}^{(i)}) \cup \mathcal{F}_{rest}^{(i)} $ . That is, $\mathcal{F}^{(i+1)}$ is the result of first removing $\mathcal{T}^{(i)}$ from $\mathcal{F}^{(i)}$ and then adding $\mathcal{F}_{(rest)}^{(i)}$ to $\mathcal{F}^{(i)}$. Note that,  $(\mathcal{F}^{(i)} \setminus \mathcal{T}^{(i)})$ and $\mathcal{F}_{rest}^{(i)}$ do not share any nodes and both these graphs are forests consisting of trees with size $\geq (1/10) \cdot  \log n / \log \log n$.
Therefore, $\mathcal{F}^{(i+1)}$ is also a forest consisting of reverse trees of size at least $(1/10) \cdot \log n / \log \log n$ only. Furthermore, $\mathcal{F}^{(i+1)}$ has at least
\[
	(\sqrt{n} - i \log^7 n) - m_i + (m_i -\log^ 7 n) = n - (i+1) \log^7 n
\]
nodes. The first $m_i$ is the size of $\mathcal{T}^{(i)}$ and $(m_i - \log^{7} n)$ is the size of $\mathcal{F}_{rest}^{(i)}$ according to \cref{obs:tree-cut}. Summarizing, after step $i$, we cut out a subtree $\mathcal{T}_{cut}^{(i)}$ of $\mathcal{F}^{(i)}$. Additionally, we  generated a forest $\mathcal{F}^{(i+1)}$ -- which is a subgraph of $\mathcal{F}^{(i)}$ and does not contain any nodes of $\mathcal{T}_{cut}^{(i)}$ -- consisting of $n - (i+1) \log^7 n$ nodes in trees of size at least $(1/10) \cdot  \log n / \log \log n$. This way, our approach may be repeated also in step $(i+1)$ and allows us to harvest additional subtrees from $\mathcal{F}^{(i+1)}$. 

The above argument can easily be translated into an induction, where we start with $\mathcal{F}^{(0)} = \mathcal{T}_{R}$. After $i^* = \sqrt{n} / \polylog n$ many steps, the set $\emph{Trees}$  contains at least $\sqrt{n} / \polylog n$ many reverse subtrees of $\mathcal{T}_{R}$ that do not share any nodes and have size $(1/10) \cdot  \log n / \log \log n$. We now set $\mathcal{F}_{v} = \{ d\}$ for any node $v$ that lies on a subtree in $\emph{Trees}$. This requires $\sqrt{n} / \polylog n \cdot (1/10) \cdot  \log n / \log \log n$ edge failures in total. Fix now  such subtree $\mathcal{T}_{R}' \in \emph{Trees}$ . In the proof of \cref{lem:case3-1} we established that, $\mathcal{T}_{R}'$ appears in $\Galpha$ with probability at least $n^{-2/5}$. As all our subtrees in $\emph{Trees}$  do not share any nodes, a similar argument as in the proof of \cref{lem:case3-1} yields that, \whp, one of our $\sqrt{n} / \log n$ many trees in $\emph{Trees}$  will appear in $\Galpha$ \whp. The root of this reverse tree $r'$ then receives $(1/10) \cdot  \log n / \log \log n$ load.
\end{proof}

\begin{lemma}
\label{lem:general-graph-extension}
Let $\mathcal{G}$ be a graph with $n$ nodes and $\mathcal{P}$ a destination-based local failover routing protocol. Then, a protocol $\mathcal{P}_K$ that operates in the clique $K_n$ can be constructed such that the following holds: For all-to-one routing to any destination $d$ and given the same set of edge failures $\mathcal{F}$, the load distribution of nodes in $\mathcal{G}$ and $K_n$ is the same.
\end{lemma}
\begin{proof}
    Consider a fixed input graph $\mathcal{G} = (V,E)$ and let for $v \in V$ denote by $\Gamma(v)$ the neighborhood of $v$ in $\mathcal{G}$. Let $K_n$ denote the clique consisting of the nodes in $V$. Consider a destination-based failover protocol $\mathcal{P} = \{\Dist{v}{\mathcal{F}_v}{d} ~|~ v \in V, \mathcal{F}_v \subseteq \Gamma(v), d \in V\}$ operating in $\mathcal{G}$. We will  transform this protocol into a protocol $\mathcal{P}_K = \{D_K(v,\mathcal{F}_v, d) ~|~ v \in V, \mathcal{F}_v \subseteq V \setminus \{v\}, d \in V\}$ that can be used in $K_n$ as follows. For each $v \in V$, $\mathcal{F}_v \subseteq \Gamma(v)$, $d \in V$ we construct the distribution $D_K(v,\mathcal{F}_v, \Gamma(v))$, where for any $w \in V$ we have  
    \begin{align*}
        \Pr \left[ D_K(v, \mathcal{F}_v, d) = w \right] = 
        \begin{cases}
            \Pr \left[ D(v, \mathcal{F}_v, d) = w \right] &\text{ if $(v,w)$ is in $E$ of  $\mathcal{G}$}\\
            0 &\text{ otherwise. }
        \end{cases}
    \end{align*}
    Intuitively this prevents edges that are only present in $K_n$ but not $\mathcal{G}$ from being used.
    As in $K_n$ each node has a larger neighborhood of $V \setminus \{v\}$ we need to define further distributions $D_K$ to complete $\mathcal{P}_K$. That is, for $v,d \in V$ and $\mathcal{F}_v \subseteq V \setminus \{v\}$ but $\mathcal{F}_v \not \subseteq \Gamma(v)$ we define the remaining distributions $D_K(v, \mathcal{F}_v, d) := D_K(v,\mathcal{F}_v \cap \Gamma(v), d)$.
    
    \medskip
    
    Assume now that the protocol $\mathcal{P}_K$ is employed in $K_n$ and consider a fixed set of edges $\mathcal{F}^{(K)}$ that are failed in $K_n$. We will show that, when operating $\mathcal{P}$ in $\mathcal{G}$ and given the failures $\mathcal{F} = E \cap \mathcal{F}^{(K)}$, then the path any packet takes is the same in both $K_n$ and $\mathcal{G}$.
    To that end, consider a packet with destination $d$ arriving at a node $v$ in both networks.
    The relevant entry for forwarding the packet in the routing table is $\alpha_K(v) \sim D_K(v, \mathcal{F}_v^{(K)},d)$ in the $\mathcal{K}_n$, and  $\alpha(v) \sim D(v, \mathcal{F}_v,d)$ in $\mathcal{G}$. Per  definition of $\mathcal{P}_K$ we have for $w \in V$ with $(v,w) \in E$ that
    \begin{align*}
         \Pr [ D_K(v, \mathcal{F}_v^{(K)},d) = w] &= \Pr [D_K(v, \mathcal{F}_v^{(K)} \cap \Gamma(v),d) =w]  \\
         &=\Pr[D_K(v, \mathcal{F}_v,d) = w ] = \Pr [D(v, \mathcal{F}_v, d) = w].
    \end{align*}
    And for $w \in V$ with $(v,w) \neq E$, we have $\Pr [ D_K(v, \mathcal{F}_v^{(K)},d) = w] = 0$ if $(v,w) \not \in E$. Hence, in both processes $v$ will forward the packet to any fixed node $w$ with exactly the same probability. This observation holds for any node $v \in V$. In other words, the probability for any fixed instance of the routing tables entries $\alpha_K$ and $\alpha$ is exactly the same in both processes. If all entries of $\alpha_K$ and $\alpha$ match, then any flow with destination $d$ will take exactly the same path in both $K_n$ and $\mathcal{G}$. Therefore, exactly the same load is created any node $v$.
\end{proof}

\subsection{Bipartite Graph Analysis}
\label{sec:details-bipartite}

\bipartitenoloop*
\begin{proof}\label{proof:bipartite-loop}
        Throughout the proof, we assume that the entries Additionally, we assume that all the entries $\alpha(w)$ at any node $w$ are still unknown. Each time the packet lands on a node $w$, we uncover the entry $\alpha(w)$ after which it gets forwarded to $\alpha(w)$. 
		When following our protocol, the packet starting on $u$ alternates between nodes in $V$ and $W$ until it reaches a node $r \in V$ such that $(v,d)$ is not failed, in which case the packet is forwarded immediately to $d$. Assume now that the packet currently resides on some node $w \in W(i)$. The packets next hop is decided by $\alpha(w)$. As $\alpha(w)$ is chosen u.a.r. from $\{u ~|~ u \in V(i+1) \land (w,u) \text{ not failed} \}$ it follows that $\alpha(w)$ can \emph{not} forward the packet directly to $d$ with probability at most 
		\[
			\Pr \left[\text{edge } (\alpha(w), d) \text{ is failed}\right] \leq \frac{|\{ u ~|~ u \in V(i+1) \land (u,d) \text{ failed }\}|}{|\{u ~|~ u \in V(i+1) \land (w,u) \text{ not failed} \}|} \leq \frac{I / 3}{2I / 3}  = \frac{1}{2}.
		\]
		The second step follows from the fact that at most $I/3$ edges between nodes in $V(i+1)$ and any fixed node $u$ may be failed (including the nodes $w$ and $d$).
		The above probability implies that any fixed packet reaches $d$ after $\ell$ transitions between $V$ and $W$ with probability at most $(1/2)^{\ell}$. Hence, after at most $\ell = K = C \log n$ alterations it arrives at $d$ with probability at least $1-n^{-C}$. Above approach assumes that the packet does not land on a node $s$ with an already uncovered entry $\alpha(s)$ within the first $\ell=2K$ hops. This is indeed guaranteed by the fact that the packet traverses the $K$ partitions of $V$ and $W$ in order and the amount of partitions is large enough.
\end{proof}

\markovtheorem*

\begin{proof} \label{proof:markov}
	We start by introducing some notation. For some interval $[a,b]$ and sequence $\{Y_i\}_{i \geq 0}$, we say that member $Y_i$ \emph{increases or remains} (short: i.o.r.) into $[a,b]$ iff $Y_i \leq b$ and $Y_{i+1} \in [a,b]$. Similar, we say that $Y_i$ \emph{decreases} into $[a,b]$ iff $Y_i > b$ and $Y_{i+1} \in [a,b]$. We denote by $\mathcal{H}^+_{[a,b]}(\{Y_i\}, \ell)$ the number of members $X_i$, $0\leq i\leq\ell$ that i.o.r. into $[a,b]$. Similar we denote by $\mathcal{H}^-_{[a,b]}(\{Y_i\}, \ell)$ the number of members that decrease into $[a,b]$. Finally, we denote by $\mathcal{H}_{[a,b]}(\{Y_i\}, \ell) := |\{Y_i :~ 0\leq i \leq \ell \land Y_i \in [a,b]\}| $ the number of times the interval $[a,b]$ is \emph{hit} by the first $\ell$ members of $\{Y_i\}$. Observe that this number can be bounded as follows.
	\begin{equation}
	\label{obs:hit-count}
		\mathcal{H}_{[a,b]}(\{Y_i\}, \ell) \leq \mathbf{1}(Y_0 \in [a,b]) + \mathcal{H}^+_{[a,b]}(\{Y_i\}, \ell) + \mathcal{H}^-_{[a,b]}(\{Y_i\}, \ell).
	\end{equation}
	Here $\mathbf{1}(Y_0 + [a,b])$ is an indicator which takes value $1$ iff $Y_0 \in [a,b]$, and $0$ otherwise.
	
	\medskip
	
	We now continue with the proof of the lemma. To derive the sum $\sum_{i=0}^{r} X_i$ we consider the sequences  $\{X_{2i}\}_{i \geq 0}$ and $\{ X_{2i + 1}\}_{i \geq 0}$, which contain only the even and odd elements, respectively.  This way, we can bound the sum as follows 
	\[
		\sum_{i=0}^{r} X_i \leq \sum_{i=0}^{\lceil r / 2 \rceil} X_{2i} + \sum_{i=0}^{\lceil r / 2 \rceil} X_{2i + 1}.
	\]

Additionally, we define for every positive integer $j$ the interval $I_j := [C \cdot \gamma^{j}, C \cdot \gamma^{j+1}]$ where both $C$ and $\gamma$ are constants. They are defined as follows. We set $C:=15 / (\delta^2 \cdot \min\{\psi,\phi, \psi \phi\})$ where $\delta$ is a constant subject to $0 < \delta < \min \{1, (\phi \psi)^{-1/4} - 1\}$. This $\delta$ will later on be used as multiplicative error when applying Chernoff bounds. Furthermore, we define $\gamma := 1/\sqrt{\phi \psi}$ and note that $\gamma > 1$.  The notion of $I_j$ allows us to bound, for example, the sum of even elements by counting the number of times the intervals $I_j$, $j\geq 0$, are hit. That is, 
	\[
		\sum_{i=0}^{\lceil r / 2 \rceil} X_{2i} \leq  \left( \sum_{j=0}^{\infty} \mathcal{H}_{I_j}(\{X_{2i}\}, \lceil r / 2 \rceil)  \cdot C \gamma^{j+1} \right) + C \cdot \lceil r / 2 \rceil.
	\]
	The added term $C \cdot  \lceil r /2 \rceil$ accounts for hits in the interval $[0,C)$. Combined with (\ref{obs:hit-count}) we may futher expand this to 
	\begin{equation}
	\label{eq:markov-main}
	 	\sum_{i=0}^{\lceil r / 2 \rceil} X_{2i} \leq \sum_{j = 0}^{\infty}  \left( \mathbf{1}(X_0 \in I_j) + \mathcal{H}^+_{I_j}(\{X_{2i}\}, \lceil r / 2 \rceil) + \mathcal{H}^-_{I_j}(\{X_{2i}\}, \lceil r / 2 \rceil)  \right) \cdot C \gamma^{j+1} +  C \cdot \lceil r / 2 \rceil.
	\end{equation}
	As a similar bound can also be created for the sum of odd elements and the remaining analysis is symmetric, we focus only on the sum of even elements from now on. To evaluate above bound, we start with the following observation.
	\begin{observation}
	\label{obs:increasing-hits}
		For $j \leq \log_\gamma r$ it holds with probability $> 1 - \exp(-4r)$ that 
		\[
		    \mathcal{H}^+_{I_j}(\{X_{2i}\}, \lceil r/ 2 \rceil) = \BigO{r \cdot \gamma^{-j}}.
		\]
		For $j > \log_{\gamma} r$ it holds with probability $>1-\exp(-4\gamma^j)$ that $\mathcal{H}^+_{I_j}(\{X_{2i}\}, \lceil r / 2 \rceil) = 0$.
	\end{observation}
	\begin{proof}
		We fix the Markov chain at step $2i$, i.e., we fix $X_{2i} = x$. We start by bounding the probability that $X_{2i}$ i.o.r. into $I_j$. To that end, we assume that $x \leq C \gamma^{j+1}$, otherwise this probability is $0$ because only elements above $I_j$ may not i.o.r. into $I_j$. The idea is to bound $X_{3i}$ by applying two Chernoff bounds in sequence. By \cref{drift:item3} and definition of $I_j$, we know that $\Ex{X_{2i+1}} \leq X_{2i} \cdot \phi \leq C \gamma^{j+1} \cdot \phi$. By \cref{drift:item1} we may apply Chernoff bounds to bound $X_{2i +1}$. This yields for the value $\delta$ we initially defined that 
	\begin{equation}
	\label{eq:markov-chernoff1}
		\Pr \left[X_{2i+1} > (1+\delta) \cdot  C  \gamma^{j+1} \cdot \phi \right] \overset{\text{Chern.}}{\leq} \exp(-C\phi  \gamma^{j+1} \delta^2 / 3 ) \leq \exp(-5\gamma^{j+1}).
	\end{equation}
	The second step follows from the definition of $C$ and $\delta$.
Therefore, we have with probability at least $1 - \exp(-5\gamma^{t+1})$  that $X_{2i+1} \leq (1+\delta)C \gamma^{j+1}\phi$. In the following, we condition on the fact that  $X_{2i+1}$ indeed follows this upper bound. In such a case, we know by \cref{drift:item2} that $\Ex{X_{3i}} \leq X_{2i+1} \cdot \psi \leq (1+\delta) C  \gamma^{j+1} \phi \psi$. We again apply Chernoff bounds, which for the same multiplicative error $\delta$ yields 
	\begin{equation}
	\label{eq:markov-chernoff2}
		\Pr \left[X_{3i} > (1+\delta)^2 \cdot  C  \gamma^{j+1} \cdot \phi \psi \right] \overset{\text{Chern.}}{\leq} \exp(-C\phi \psi  \gamma^{j+1} \delta^2 (1+\delta) / 3) \leq \exp(-5 \gamma^{j+1}).
	\end{equation} 
	The second step again follows from the definition of $C$ and $\delta$.
	In case the bad events bounded in (\ref{eq:markov-chernoff1}) and (\ref{eq:markov-chernoff2}) do \emph{not} occur, we have that $X_{3i}\leq (1+\delta)^2 \cdot C \gamma^{j+1} \phi \psi$. We established that the probability for this is at least 
	\[
		1 - (\exp(-5\gamma^{t+1} ) + \exp(- 5\gamma^{j+1})) =1 -  2\exp(- 5\gamma^{j+1}).
	\]
	 Now, when first using that $\gamma = 1/ \sqrt{\phi \psi}$ in the first step and then $\delta < (\psi \phi)^{-1/4} - 1$ in the second, we get 
	\[
		 (1+\delta)^2 \cdot C \gamma^{j+1} \phi \psi = (1+\delta)^2 C  \gamma^{j} \sqrt{\psi \phi}  < C  \gamma^{j}
	\]
	Therefore, $X_{3i}$ does not lie in $I_j$ with probability at least $1 - 2\exp(-5\gamma^{j+1})$. In other words, we showed that a fixed element $X_{2i}$ does i.o.r. into $I_j$ with probability at most $2 \exp(-5\gamma^{j+1})$.
	
	\medskip
	
	In total, we consider the first $\lceil r / 2 \rceil < r$ many elements of $\{X_{2i}\}_{i \geq 0}$. Using the result of the previous paragraph, we majorize the number of them which i.o.r. in $I_j$ 
	by $ B \sim \Bin(r~,~ 2\exp(- 5\gamma^{j+1}))$. In case $j >  \log_{\gamma} r$, it is easy to see that $B = 0$ occurs with good probability. More precisely, \cref{obs:increasing-hits} follows as 
	\[
	    \Pr [B = 0] \geq  (1 - 2\exp(-5 \gamma^{j}))^{r} \geq 1 - 2r \exp(-5 \gamma^{j}) \geq 1 - \exp(-4 \gamma^j).
	\]
	In the last step, we used that $2r< \exp(r) < \exp(-\gamma^{j})$.
	For smaller values of $j$, we use the PDF of the binomial distribution to derive for $0 \leq k \leq r$
	\[
		\Pr \left[B = k \right] \leq \binom{r}{k} \cdot 2^k \exp(-5\gamma^{j+1} \cdot k) \leq \left( \frac{e r}{k}\right)^k \cdot 2^k \exp(-5\gamma^{j+1}k).
	\]
	For values of $k \geq  2e \cdot r \cdot \gamma^{-(j+1)} = \Omega (r \cdot \gamma^{-j})$ the right-hand side can be further simplified when applying this bound to the $k$ in the denominator:
	\begin{align*}
		\Pr \left[B = k \right] \leq  (\gamma^{j+1} / 2)^k \cdot 2^k \exp(-5\gamma^{j+1} k) = \left( \frac{\gamma^{j+1}}{\exp(5\gamma^{j+1})} \right)^k < \left( \frac{5\gamma^{j+1}}{\exp(5\gamma^{j+1})} \right)^k  \\
		\leq \exp(-5\gamma^{j+1} \cdot k / 2) 
		\leq \exp(-5e \cdot r^2) <  \exp(-5r)
	\end{align*}
    To derive the first term in the second line, we used that $x / \exp(x) < \exp(x/2)$ holds for $x \geq 0$, which we applied for $x = 5\gamma^{j+1}$. 
	As the bound  $\Pr[B =k] < \exp(-5r)$ holds for any value of $k  \geq   2e \cdot r \cdot \gamma^{-(j+1)} $ and $k$ may take at most $r$ different values, we get that
	\[
		\Pr \left[ B \geq 2e \cdot r \cdot \gamma^{-(j+1)} \right] \leq r \exp(-5r) \leq \exp(-4r)
	\]
	and the result of \cref{obs:increasing-hits} follows.
	\end{proof}
	Hence, we established that, \whp,  $ \mathcal{H}^+_{I_j}(\{X_{2i}\}, \lceil r / 2 \rceil)$ shrinks geometrically with increasing $j$ until it eventually reaches $0$ for large values of $j$. Remember, the other important type of hits w.r.t. $I_j$ are caused by decreasing members. Note that each decrease of some member $X_{2i}$ into $I_j$ needs to be preceded by one of the following events (i) some member $X_{2i'}$ with $i' < i$ must have increased into an interval $I_{j'}$ with $j' > j$, or (ii) $X_0$ was already above $C \gamma^{j+1}$ (enabling the first decreasing member). Therefore, 
	\begin{align}
	\label{eq:decreasing-hits}
		 \mathcal{H}^-_{I_j}(\{X_{2i}\}, \lceil r / 2 \rceil) \leq  \mathbf{1}(X_0 > C \gamma^{j+1}) +  \sum_{j' > j} \mathcal{H}^+_{I_{j'}}(\{X_{2i}\}, \lceil r / 2 \rceil).
	\end{align}
	In the following, we assume the event in which \emph{all} intervals $I_j$, $j \geq 0$, follow the bound given by \cref{obs:increasing-hits} holds. By a union bound application, the probability for this is at least
	\[
	     \geq  1 - \left(\log_{\gamma} r \cdot \exp(-4r) + \sum_{j > \log_{\gamma} r} \exp(-4\gamma^{j}) \right) > 1 - \exp(- 3r).
	\]
	To derive the result on the right-hand side we assume $r > C_{\phi \psi}$, where $C_{\phi \psi}$ is a sufficiently large constant that depends on $\gamma$.
	Conditioned on the aforementioned event,  we may use the result of \cref{obs:increasing-hits} bound the number of i.o.r with the help of a geometric series: $\sum_{j' > j} \mathcal{H}^+_{I_j'}(\{X_{2i}\}, \lceil r / 2 \rceil) \leq \sum_{j' > j} \BigO{r \cdot \gamma^{-j'}} = \BigO{r \cdot \gamma^{-j}}$. We combine this result with (\ref{eq:decreasing-hits}) which yields 
	\begin{align*}
		\mathcal{H}^-_{I_j}(\{X_{2i}\}, \lceil r / 2 \rceil)	\leq  \mathbf{1}(X_0 > C \gamma^{j+1})  +  \begin{cases}
			\BigO{r \cdot \gamma^{-j}}  & \text{ if $j \leq  \log_\gamma r$} \\
			0 & \text{ otherwise.}
		\end{cases}
	\end{align*}
	We are now ready to further simplify (\ref{eq:markov-main}). As we can bound the number of increasing and decreasing hits each by $\BigO{r \cdot \gamma^{-j}}$, we have for some fixed interval $j \leq \log_\gamma r$ that 
	\begin{align}
		\left( \mathbf{1}(X_0 \in I_j) + \mathcal{H}^+_{I_j}(\{X_{2i}\}, \lceil r / 2 \rceil) + \mathcal{H}^-_{I_j}(\{X_{2i}\}, \lceil r/ 2 \rceil)  \right) \cdot C \gamma^{j+1} \\
\leq  \mathbf{1}(X_0 \geq C \gamma^{j}) \cdot C \gamma^{j+1} + \BigO{r}. \label{eq:markov-11}
	\end{align}
	Note, in case of $j > \log_\gamma r$, the bound in (\ref{eq:markov-11}) becomes just $\mathbf{1}(X_0 \geq C \gamma^{j}) \cdot C \gamma^{j+1}$ as, according to \cref{obs:increasing-hits}, no i.o.r. hits occur  and $\mathcal{H}^-_{I_j}(\{X_{2i}\}, \lceil r / 2 \rceil) \leq \mathbf{1}(X_0 > C \gamma^{j+1})$. When plugging these inequalities into (\ref{eq:markov-main}) we get that, 
	\begin{align*}
		\sum_{i=0}^{\lceil r/ 2 \rceil} X_{2i} \leq C \cdot \lceil r /2 \rceil + \sum_{j=0}^{\infty} \mathbf{1}(X_0 \geq C \gamma^{j+1}) \cdot C\gamma^j+ \sum_{j = 0}^{\log_\gamma ( r) - 1}    \BigO{r} \\
		=  \BigO{ r } +  \sum_{j=0}^{\infty} \mathbf{1}(X_0 \geq C \gamma^{j}) \cdot C\gamma^{j+1} + \BigO{\log (r) \cdot r} + 
	\end{align*}
	Observe that $X_0 = \BigO{r}$ by assumption of \cref{lem:markov}. Therefore $\sum_{j=0}^{\infty} \mathbf{1}(X_0 \geq C \gamma^{j}) \cdot C\gamma^{j+1} \leq \BigO{\log(r) \cdot r}$.
	
	\medskip
	
	What remains is the proof for the sum of the odd members $\sum_{i=0}^{\lceil r / 2 \rceil} X_{2i + 1}$. A repetition of the current proof, while exchanging $\psi$ with $\phi$ and $\{X_{2i}\}$ with $\{X_{2i + 1}\}$, yields with probability $> 1- \exp(-3r)$ 
	\begin{align*}
		\sum_{i=0}^{\lceil r/ 2 \rceil} X_{2i + 1} \leq C \cdot \lceil r /2 \rceil + \sum_{j=0}^{\infty} \mathbf{1}(X_1 \geq C \gamma^{j+1}) \cdot C\gamma^j+ \sum_{j = 0}^{\log_\gamma (r) - 1}    \BigO{r} \\
		= \BigO{\log (r) \cdot  r } +  \sum_{j=0}^{\infty} \mathbf{1}(X_1 \geq C \gamma^{j}) \cdot C\gamma^{j+1}
	\end{align*}
	In this setting we cannot immediately derive $X_1 = \BigO{r}$ from the assumptions of \cref{lem:markov} to bound the infinity sum. However, as $X_0 = \BigO{r}$ it follows by  \cref{drift:item2} that  $\Ex{X_1} = \BigO{r}$. By \cref{drift:item1} we may apply Chernoff bounds, which easily yields that  $X_1 \leq C' \cdot r$ with probability $\geq 1- \exp(-4r)$ for some large enough constant $C' > 0$. Just as in the case of even members, this allows us to bound the infinite sum by  $\BigO{\log(r) \cdot r}$. Summarizing, we showed that, both, the sum of members with even and odd indices sum up to $\BigO{\log (r) \cdot r}$ with probability at least $1- \exp(-3r)$. The result follows by adding up the odd and even elements and applying a union bound.
\end{proof}

\bipartitemarkov*

		\begin{proof} \label{proof:bipartite-markov}
		Assume all entries $\alpha(v)$ to determine the sets $L_0, L_1, ..., L_{i-1}$ were already uncovered. Given $L_{i-1}$, we will now inductively determine $|L_{i}|$. To that end, we distinguish two cases that depend on  the parity of $i$,
		
		\case{1} $i \geq 1$ is odd. In this case, the nodes in $L_{i-1}$ belong to $V(K - \lceil i/2 \rceil)$ and nodes in $L_i$ will lie in $W(K - \lceil i/2 \rceil - 1)$. More precisely we have that $L_i = \{ w ~|~ w \in W(K - \lceil i /2 \rceil - 1) \land \alpha(w) \in L_{i-1}\}$. Remember, for a fixed $w \in  W(K - \lceil i /2 \rceil - 1)$, the entry $\alpha(w)$ is chosen u.a.r. from $V(K - \lceil i /2 \rceil) \setminus \mathcal{F}_w$ according to the definition of $\mathcal{P}_{B}$. Therefore, $\alpha(w)$ lies in $L_{i-1}$ with probability at most 
		\[
			\Pr \left[ \alpha(w) \in L_{i-1} \right] \leq \frac{|L_{i-1}|}{|V(K - \lceil j /2 \rceil) \setminus \mathcal{F}_w|} \leq \frac{3}{2}  \cdot \frac{|L_{i-1}|}{I}
		\]
		The second step follows when using the fact that $\mathcal{F}_w$ can contain at most $I/3$ nodes of $V(K - \lceil j /2 \rceil)$ due to our constraints on the set of failed edges in \cref{thm:bipartite}.
		Note that the entry $\alpha(w)$ is determined independently from other nodes $w \in W(K - \lceil i /2 \rceil - 1)$. Therefore, $|L_i|$ can be modeled by a sum of $|W(K - \lceil i /2 \rceil - 1)|$ Poisson trials and we have 
		\[
			\Ex{|L_i|} \leq \sum_{w \in W(K - \lceil i /2 \rceil - 1)} \frac{3}{2}  \cdot \frac{|L_{i-1}|}{I} = I \cdot \frac{3}{2}  \cdot \frac{|L_{i-1}|}{I} = \frac{3}{2} \cdot |L_{i-1}|.
		\]

\case{2} $i \geq 2$ is even. In this case, we have $L_{i-1} \subseteq W(K  -  i/2 -1)$ and $L_{i}$ will consist of nodes in $V(K - i/2 -1)$. More preceisely, we have that $L_i = \{v ~|~ v \in V(K - i/2 -1) \land v \not \in V_G \land \alpha(v) \in L_{i-1} \}$. This case differs from the previous as we only need to consider the entries $\alpha(v)$ for $v \not \in V_G$. This is because if $v \in V_G$, then $v$ can forward its flows directly to $d$. Nodes $v \in V(K - i /2 -1) \setminus V_G$ again chooses the entry $\alpha(v)$ u.a.r. from $W(K - i/2 -1) \setminus \mathcal{F}_{v}$. For such a node -- similar to the previous case -- it holds that 
\[
	\Pr \left[ \alpha(v) \in L_{i-1} \right] \leq  \frac{|L_{i-1}|}{|W(K - i/2 -1) \setminus \mathcal{F}_{v}|} \leq \frac{3}{2}  \cdot \frac{|L_{i-1}|}{I}
\]
 As in the previous case, we can also model $|L_i|$ by a sum of Poisson trials. This time, the amount of such trials is $|\{v ~|~ v \in V(K - i/2 -1) \land v \not \in V_G\}|$. Note, due the constraints on the failures in \cref{thm:bipartite} at most $I/3$ nodes in $V(K - i/2 -1)$ can have their link to the destination $d$ failed. Hence, the cardinality of the aforementioned set can be bounded by $I/3$ and 
\[
	\Ex{|L_i|} \leq \frac{I}{3} \cdot \frac{3}{2} \frac{|L_{i-1}|}{I} = \frac{1}{2} |L_{i-1}|.
\]
	\end{proof}

\bipartitelowerbound*

\begin{proof}\label{proof:bipartite-lower-bound}
    	We consider a fixed failover protocol $\mathcal{P}$ together with a destination node $d$. Assume that the failover entries $\alpha$ of all nodes have already been decided by drawing them from the distributions of $\mathcal{P}$. Consider some node $v_1 \in V$. In case the edge $(v_1,d)$ is failed, it will forward a packet with destination $d$ according to $\alpha(v_1, \{d\},d)$, which is some node $w_1 \in W$. If now $\mathcal{F}_{w_1} = \emptyset$, that is no edge incident to $w_1$ is failed, then $w_1$ will again forward the packet to some node $v_2 \in V$ according to $\alpha(w_1,\emptyset,d)$. We repeat this approach and track the packet for $m := (1/3) \cdot \log n / \log \log n$ alternations between $V$ and $W$. We observe that the packet starting at $v_1$ travels the path $(v_1,w_1, v_2, w_2 , .... v_m, w_m)$ while still not arriving at the destination, in case the edges $(v_1,d), (v_2,d), ...., (v_m,d)$ each are failed and no edges incident to nodes in $w\in W$ are failed.

        As the nodes $v_1, ... , v_m$ are determined through the routing table $\alpha$, we cannot directly predict these nodes and fail the corresponding edges $(v_i,d)$. However, we can fail a  subset of edges $\mathcal{F} \subseteq \{ (v,d) ~|~ v \in V \}$  of size $|\mathcal{F}| = \varepsilon \cdot n / \log n$ that we select uniformly at random. Ideally we want this set of failures to contain all edges in the set $\{(v_1,d), ... , (v_m,d)\}$. The probability for this can be modelled with a hypergeometric distribution and is 
    \[
    	\geq \frac{ \frac{\varepsilon n}{\log n}}{ n} \cdot \frac{\frac{\varepsilon n}{\log n} - 1}{n}  \cdot  ...  \cdot \frac{\frac{\varepsilon n}{\log n}-m}{n} \geq \left( \frac{\frac{\varepsilon n}{2\log n}}{n} \right)^m = \left( \frac{\varepsilon/2}{\log n} \right)^m = \left( \frac{1}{n}\right)^{(1/3) + o(1)}.
    \]
    Hence, we established that, by failing $\varepsilon \cdot n / \log n$ edges of the form $(v,d)$ uniformly at random, we can create a path of length $2m$ with probability at least $n^{-1/3 - o(1))}$. In such a case, the flow initiated at $v_1$ travels for $>2m$ hops and hits $>m$ nodes at $V$, each of which also initiate one flow. This implies that $\mathcal{L}(v_m) > m$. \emph{In expectation}, we have that at least $|V| \cdot n^{-1/3 - o (1)} > \sqrt{n}$ nodes are each part of such a path of length at least $2m$. Even if all these paths terminate in the same node $w \in V$, this implies that in expectation there is at least one node $w \in V$ with $\mathcal{L}(w) = \Omega(\log n / \log \log n)$.
\end{proof}

 \subsection{Clos Topology Routing Example}
 \label{sec:clos-example}
 
 In order to illustrate the behavior of our routing protocol given in \cref{def:interval-routing}, we present an exemplary Clos topology in \cref{fig:clos-example1}. 
 As described in \cref{sec:clos-definition}, the nodes are partitioned into blocks, which then are again split into clusters and finally intervals. Note that, in our example, only level $\ell = 0$ has more then one cluster. Additionally, each node on level $\ell = 2$ lies in its own block. These blocks were omitted to improve visual clarity. 
 The goal is to route the packet, currently residing at node $s$, towards destination node $d$ with sequence $S_d = (3,6)$. As described in \cref{sec:clos-definition}, the nodes are partitioned

First (see a) and b) in \cref{fig:clos-example1}), the packet is forwarded via \ruleone to random nodes in an interval of a block that lies one level lower. This is done until the packet resides on a node in $\bl{\emptyseq}$. This is the first block the packet reaches such that the block's sequence (in this case $\emptyseq$) is a prefix of $S_d$. In the next step, the packet needs to be forwarded via \ruletwo into $T(3)$, the fat subtree containing $d$, to get closer to the destination $d$. In our example, we assume that after two hops the random choices caused the packet to land on a node $v$ in cluster $\cl{\emptyseq}{2}$ of $\bl{\emptyseq}$. According to the definition of the Clos topology, each node inside this cluster only has a single link into $T(3)$, all of which are incident to $d_{1,2}$.
In case no links are failed, the packet would simply be forwarded first to $d_{1,2}$ and then finally towards $d$. However, in our example we assume that the link $(v,d_{1,2}$) is failed and continue our example  with \cref{fig:clos-example2}.

On the left-hand side in \cref{fig:clos-example2} we take a closer look at $\cl{\emptyseq}{2}$ and its links into level $\ell=1$. According to the Clos topology definition, the nodes in $\cl{\emptyseq}{2}$ and the second node in each block $\bl{i}$ form a complete bipartite graph. These "second nodes" (denoted by $\vcl{\emptyseq}{2}$) are partitioned into vertical intervals, separated by the orange lines in the image. The idea is now to forward the packet, currently residing on $v$, towards $d_{1,2}$ just as in the protocol for the bipartite graph in \cref{sec:bipartite}. That is, until the packet reaches a node $v'$ such that the link to $d_{1,2}$ is not failed, it "ping-pongs" between intervals of $\cl{\emptyseq}{2}$ and vertical intervals of $\vcl{\emptyseq}{2}$. Our routing protocol implements this as the forwarding rule applied to the packet alternates between \ruletwo and \ruleone.

After being forwarded to $d_{1,2}$, we continue our example with the right-hand side of \cref{fig:clos-example2}. The packet now resides in $T(3)$ and is close to the destination. If the link $(d_{1,2},d)$ is not failed, then the packet is forwarded directly to $d$ via \ruletwo. However, we assume that this is not the case. Just as in subgraph considered in the left-hand side of \cref{fig:clos-example2} one can again observe that $T(3)$ is a complete bipartite graph. This allows us to again apply ideas from the bipartite routing protocol. The packet "ping-pongs" via \ruleone and \ruletwo. until it hits the first node in $B(3)$ that can reach $d$ directly.

 \begin{figure}
    \centering
    \includegraphics[scale=0.85,page=2]{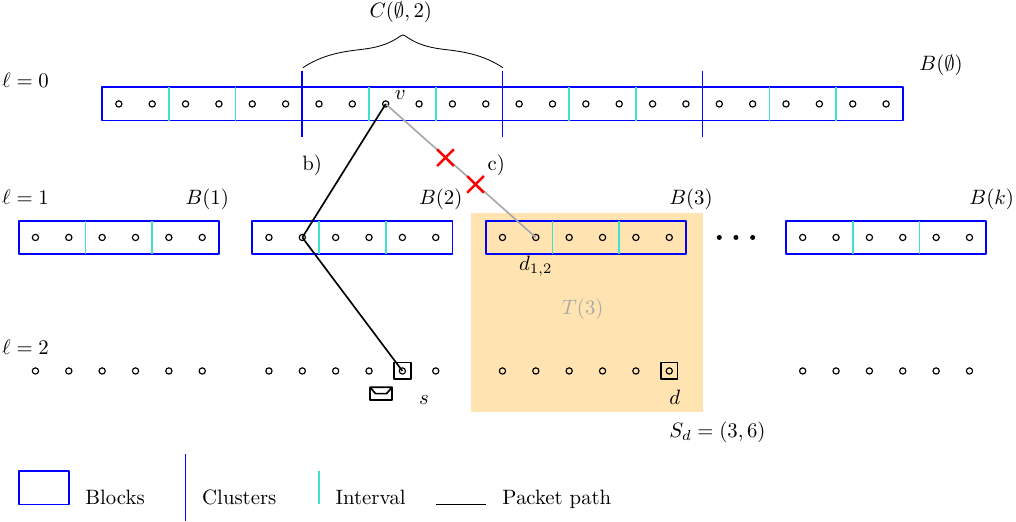}
    \caption{Clos topology with $k=12$, $L=2$ and interval size $I=2$. For easier visibility, edges are not visible. A packet residing at some bottom node of layer $s$ is forwarded towards destination $d$ with sequence $S_d = (3,6)$ until it encounters a failed edge.}
    \label{fig:clos-example1}
\end{figure}

\begin{figure}
    \centering
    \includegraphics[scale=0.88,page=5]{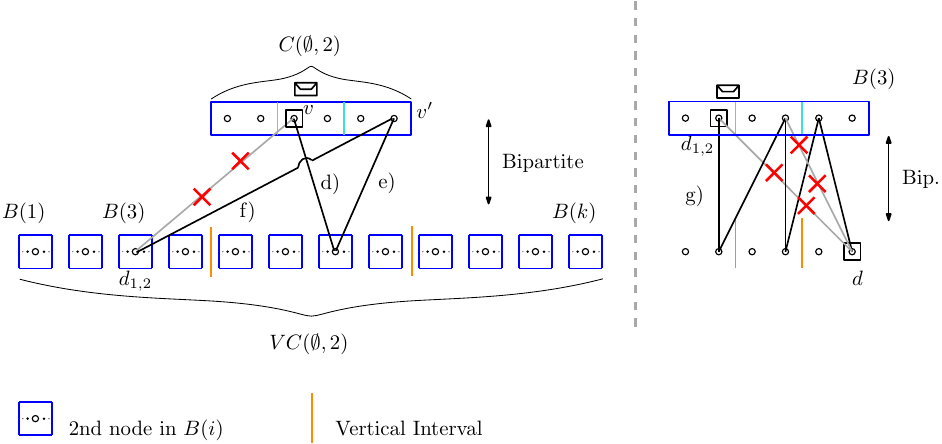}
    \caption{On the left, a closer look at $\cl{\emptyseq}{2}$. The nodes in $\cl{\emptyseq}{2}$ form a complete bipartite graph together with the second node in each block $\bl{i}$. On the right, we have $T(3)$. After landing on $d_{1,2}$, the packet cannot be forwarded directly to $d$ as we assume the link $(d_{1,2},d)$ is failed.}
    \label{fig:clos-example2}
\end{figure}

\subsection{Clos Topology Analysis}
\label{sec:details-fat-tree}

\atreeinduction*

\begin{proof} \label{proof:a-tree-induction}
    We will show the statements by induction. For the base case of $\ell = 0$ the statements follow directly as all entries are still covered and each node starts with $1$ flow.
    Next to the inductive step. We assume that the statement of \cref{lem:a-tree-induction} holds for some fixed $\ell \geq 0$. That is, we assume all edges outside $\tree{\prefix{S_d}{\ell}}$ to be uncovered and, most importantly, that on each node in $\bl{\prefix{S_d}{\ell}}$  at most $\BigO{k^{\ell}}$ flows are stopped. We will now establish the statements of \cref{lem:a-tree-induction} w.r.t. $\ell + 1$. We need to uncover entries of nodes that lie in $\tree{\prefix{S_d}{\ell}}$ but not in $\tree{\prefix{S_d}{\ell+1}}$, and show that each node in $\bl{\prefix{S_d}{\ell + 1}}$ receives $\BigO{k^{\ell + 1}}$ load. 
    We do this, by uncovering edges in two steps. To that end, we define $S_\ell := \prefix{S_d}{\ell}$ to ease readability and also define $\mathcal{B}_\ell$ as the set of  children of block $\bl{S_\ell}$ \emph{excluding} $\bl{\prefix{S_d}{\ell+1}}$. Similar we let $\mathcal{T}_\ell$ denote the corresponding trees, each rooted in one of the blocks in  $\mathcal{B}_\ell$. 
    
    In the first step, we uncover entries of nodes in $\mathcal{T}_\ell$, track their unleashed flows until they hit some node in layer $\ell +1$ and show in \cref{obs:a-tree-1} that the nodes of each vertical cluster in $\mathcal{B}_\ell$ receive $\BigO{k}$ load. 
    
    In the second step, we uncover the remaining entries, which belong to nodes in $\bl{S_\ell}$ and $\mathcal{B}_\ell$. The nodes in $\bl{S_\ell}$ each belong to one of the clusters $\cl{S_\ell}{i}$, $i \geq 1$, and the nodes in $\mathcal{B}_\ell$ to some vertical cluster $\vcl{S_\ell}{i'}$, $i' \geq 1$. For any fixed $i\geq 1$, the nodes in $\cl{S_\ell}{i}$ and $\vcl{S_\ell}{i}$ form a complete bipartite graph. Nodes in this graph will forward their flows over $d_{s+1,i}$, which is the $i$-th node in $\bl{S_\ell + 1}$, as illustrated in \cref{fig:ping-pong}. By the induction hypothesis, we know that each node in $\bl{S_\ell}$ hosts $\BigO{k^{\ell}}$ flows \whp, implying that $\cl{S_\ell}{i}$ contains $\BigO{k^{\ell + 1}}$ flows in total. Additionally, at the end of the first step, we know by \cref{obs:a-tree-1} that each cluster $\vcl{S_\ell}{i}$ contains $\BigO{k}$ flows in total \whp. Therefore, the node $d_{s+1,i}$ will receive $\BigO{k} + k \cdot \BigO{k^{\ell}} =  \BigO{k^{\ell + 1}}$ many flows after the entries in $\bl{S_\ell}$ and $\mathcal{B}_\ell$ are uncovered. As $\bl{S_{\ell+1}} = \bl{\prefix{S_d}{\ell+1}} = \{d_{\ell+1,i} ~|~ i \geq 1\}$ this implies the statement in \cref{item:a-tree-3} of \cref{lem:a-tree-induction}. 
    
    Note that, when uncovering edges in this second step, some flows will "ping-pong" between $\cl{S}{i}$ and $\vcl{S}{i}$ until they reach $d_{\ell+1,i}$. This could cause a forwarding cycle. We present \cref{obs:a-tree-hops} which states that packets "ping-pong" $<2K$ times until hitting $d_{\ell+1,i}$ \whp, which prevents them from travelling in a cycle as this would require them to traverse at least $2K$ many intervals. Therefore, also the statement in \cref{item:a-tree-1} is fulfilled. The statement in \cref{item:a-tree-2} also follows from a combination of \cref{obs:a-tree-1} and \cref{obs:a-tree-hops} as all flows from nodes in $\mathcal{T}_\ell$ first reach some node in $\mathcal{B}_\ell$ and later on are forwarded to some $d_{\ell+1,i}$ without being trapped in a forwarding loop.
    
    Summarizing, we established the statements in \cref{item:a-tree-1,item:a-tree-2,item:a-tree-3} w.r.t. $\ell + 1$. Throughout our arguments, we assumed that (i) \cref{item:a-tree-1,item:a-tree-2,item:a-tree-3} hold with respect to $\ell$ (induction hypothesis), (ii) the event of \cref{obs:a-tree-1} holds, and (iii) the event of \cref{obs:a-tree-hops} holds for any cluster $i \geq 1$. By a union bound application the probability for this is at least $1 - (4\ell \cdot k^{-4} + k^{-\omega(1)} + k^{L-1} \cdot 2k^{-(3 + L)}) > 1 - 4(\ell + 1) k^{-4})$. The factor $k^{L-1}$ originates from the fact that no block contains more clusters than $\bl{\emptyseq}$, which contains exactly $K^{L-1}$ many clusters.
    In the remainder of this proof, we list \cref{obs:a-tree-1,obs:a-tree-hops} together with their corresponding proofs. \qedhere
    \begin{observation}
    \label{obs:a-tree-1}
         Assume we uncover the entries of all nodes in $\mathcal{T}_\ell$, excluding those of nodes in blocks of $\mathcal{B}_\ell$. Then, the flow of any node in $\mathcal{T}_\ell$ reaches some node in $\mathcal{B}_\ell$.  Additionally, for any $i \geq 1$, the nodes in $\vcl{S_\ell}{i} \setminus \{d_{\ell+1,i}\}$ will receive in total $\BigO{k}$ flows with probability at least $1-k^{-\omega(1)}$.
    \end{observation}
    \begin{proof} \label{proof:a-tree-one}
        We start by showing that the nodes in a fixed vertical cluster $\vcl{S_\ell}{i} \setminus \{d_{\ell,i}\}$ will $\emph{not}$ receive more than $\BigO{k}$ load \whp.
         Let $R$ be the set of nodes in  $\vcl{S_\ell}{i} \setminus \{d_{\ell+1,i} \}$.
        Consider an arbitrary node $v \in \mathcal{T}_\ell$ together with its block $\bl{S}$. Per definition of $\mathcal{T}_\ell$ it follows that $\bl{S}$ is a successor of $\bl{S_\ell}$ that does not share a prefix with $\mathcal{S}_d$. Therefore, such a node $v$ will follow \ruleone of our routing protocol (see \cref{def:interval-routing}).
        After uncovering the entry of $v$, its flows will be directed to some node in $\mathcal{T}_\ell$ that lies one level lower, on $|S| - 1$. Note that this cannot be prevented by the adversary, as it may not fail more than $I/3$ edges incident to $v$ and $v$ has $I$ links into level $|S| - 1$.
        
        Throughout this proof we denote by $L_a$ the set of nodes that (i) lie on level $\ell + 1+a$ (ii) belong to  $\mathcal{T}_\ell$ and (iii) forward their flows to some node in $L_{a-1}$ after their entries have been uncovered. To complete this inductive definition, we define $L_0$ to be $R$.
        As all nodes in $\mathcal{T}_\ell$ forward according to \ruleone, given $L_i$, we may determine $L_{i+1}$ by uncovering edges of nodes in level  $\ell+(i+1)$ only.
        Hence, after uncovering all entries of nodes in  $\mathcal{T}_\ell$ level-by-level, we may determine sizes of the sets $|L_i|$. We can then sum up these sets to determine $\sum_{i=0}^{L - \ell+1} |L_i|$, which is the load that all nodes in $R$ receive.
        We now show how, given $L_i$, the size of $|L_{i+1}|$ can be calculated. Consider a node $v \in L_{i}$ that lies in some cluster interval $\cli{S}{i}{j}$ with $|S| = \ell + 1+i$. Such a node will only receive flows from neighbors $w \in \vcli{S}{i}{(j-1) \mod K}$ in layer $\ell + 2 + i$ (above we established that in $\mathcal{T}_\ell$ only \ruleone is used for forwarding). Hence, when uncovering the entries of the nodes on level $\ell + 2+i$, such a node $w$ forwards its load to $v$ with probability at most 
        \[
            \Pr \left[\alpha(w) \in L_i \right] \leq \frac{|\{v \in L_i \land v \in \cli{S}{i}{j} \}|}{I - |\mathcal{F}_v|} \leq \frac{3}{2} \cdot \frac{|\{v \in L_i \land v \in \cli{S}{i}{j} \}|}{I}.
        \]
        
        In the second step of this calculation, we used the fact that the adversary is only allowed to fail $I/3$ edges that connect nodes from $\cli{S}{i}{j}$ with $v$. For nodes  $v \in \cli{S}{i}{j}$ there are only $I$ neighboring nodes to receive flows from, namely those in $\vcli{S}{i}{j}$. Also, note that each node sets its entry independently from other nodes. Therefore, when summing over all blocks and intervals, we get that $L_{i+1}$ can be represented by a sum of Poisson trials with expected value 
        \begin{align*}
            \Ex{|L_{i+1}|} \leq \sum_{\substack{S,i,j \\ \text{ with } \cli{S}{i}{j} \cap L_i \neq \emptyset}} \sum_{w \in \vcli{S}{i}{j}}  \frac{3}{2} \cdot \frac{|\{v \in L_i \land v \in \cli{S}{i}{j} \}|}{I} \\
            = \sum_{w \in \vcli{S}{i}{j}}  \sum_{\substack{S,i,j  \\ \text{ with } \cli{S}{i}{j} \cap L_i \neq \emptyset}} \frac{3}{2} \cdot \frac{|\{v \in L_i \land v \in \cli{S}{i}{j} \}|}{I} \\
            = \sum_{w \in \vcli{S}{i}{j}} \frac{3}{2} \cdot \frac{|L_i|}{I} = \frac{3}{2} \cdot |L_i|.
        \end{align*}
        When applying Chernoff bounds, this  yields that $|L_{i+1}| \leq 3 |L_i|$ with probability $1-\exp(-\Omega(|L_i|))$. The above approach can easily be translated into an introduction which yields that, with probability $\geq 1-i \cdot \exp(-\Omega(|L_0|))$ we have $|L_{i}| \leq 3^i |L_0|$. Initially, we established that the nodes in $L_0 = R$ receive $\sum_{i=0}^{L - \ell+1} |L_i|$ load. As the topology has at $L+1 = \BigO{1}$ levels, it follows that this number lies in $\BigO{k}$ with probability at least $ 1-L \cdot \exp(-\Omega(|L_0|) \geq  1 - k^{-\omega(1)}$.
        
        Hence, we showed that the nodes in $R = \vcl{S_\ell}{i} \setminus \{d_{\ell+1}{i}\}$ receive $\BigO{k}$ flows in total.  As there only exist $k^{\BigO{1}}$ vertical clusters w.r.t. $S_\ell$, a union bound application over all such vertical clusters yields the desired result.
    \end{proof}

    \begin{observation}
    \label{obs:a-tree-hops}
        Fix cluster $i$ of $\bl{S_\ell}$. Then, after uncovering the entries of all nodes in $\cl{S_\ell}{i}$ and $\vcl{S_\ell}{i} \setminus \{d_{\ell +1,i}\}$, the flows stopped at any node in $v \in \cl{S_\ell}{i} \cup \vcl{S_\ell}{i} \setminus \{d_{\ell +1,i}\}$ reach $d_{\ell+1,i}$ within at most $2 K$ further hops with probability $\geq 1 - 2k^{-(3+L)}$.
    \end{observation}
    \begin{proof} \label{proof:a-tree-two}
        The clusters $\cl{S_\ell}{i}$ and $\vcl{S_\ell}{i}$ form a complete bipartite graph. Therefore, the proof is similar as the proof of \cref{obs:bip-noloops}. Fix now a packet that originates on some node $v$ in some interval $\cli{S_\ell}{i}{j}$ or vertical interval $\vcli{S_\ell}{i}{j}$ for some $j\geq 1$. Until this packet reaches $d_{\ell+
        1,j}$ it "ping-pongs" between intervals on level $\ell$ and $\ell+1$ (see \cref{fig:ping-pong}). Each time it hits a node $w \in  \vcli{S_\ell}{i}{j'}$ for some $j'$, it has the chance on the next hop to hit a node $u \in \cli{S_\ell}{i}{j' + 1 \mod K}$ that has its edge to $d_{\ell+1,i}$ not failed. As the entry $\alpha(w)$ is chosen uniformly at random out of $\cli{S_\ell}{i}{(j' + 1) \mod K} \setminus \mathcal{F}_w$, it follows that the probability for this is at least
        \[
            1 - \frac{I/3}{\cli{S_\ell}{i}{(j' + 1) \mod K} \setminus \mathcal{F}_w|} \geq 1 - \frac{I/3}{2I/3} = \frac{1}{2}.
        \]
        This bound follows from the fact that restrictions on the failures placed by the adversary as stated in \cref{thm:a-tree-main}.
        Conversely, the packet on $v$ will ping-pong more than $x$ times with probability only $1/2^x$. For $x= K = (4+L) \log n$, this probability is less than $k^{-(4+L)}$. Besides $v$ there are $\leq 2k$ nodes in $\cl{S_\ell}{i} \cup \vcl{S_\ell}{i} \setminus \{d_{\ell+1,i}\}$. A union bound application yields that the flows starting from any such node will reach $d_{\ell+1,i}$ in $2K$ hops with probability $\geq 1 - k^{-(4 + L)} \cdot 2k = 1 - 2k^{-(3+ L)}$.
    \end{proof}
\end{proof}

    \atreelowerbound*
    
    \begin{proof} \label{proof:a-tree-lower-bound}
    Throughout the analysis, we consider a fixed destination $d$ and will focus on  $\tree{\prefix{S_d}{L-1}}$. According to \cref{obs:last-tree} this subgraph is a complete bipartite graph consisting of the two clusters $\cl{\prefix{S_d}{L-1}}{1} = \bl{\prefix{S_d}{L-1}}$ and $\vcl{\prefix{S_d}{L-1}}{1}$ of $(k/2)$ nodes each. Note that $d \in \vcl{\prefix{S_d}{L-1}}{1}$. In the following we will call these sets of nodes $V$ and $W$. Assuming we do not fail any edges, it follows that from the fact that $\mathcal{P}$ is fairly balanced (see \cref{def:fairly-balanced}) that $\mathcal{L}(v) = \Theta(k^{L-1})$ at any node $v$ of $V = \bl{\prefix{S_d}{L-1}}$.
    As $\mathcal{P}$ is a local failover protocol, this is also true if we fail edges inside of $\tree{\prefix{S_d}{L-1}}$. Because $\mathcal{P}$ must forward over shortest-paths, it follows that any flows arriving at a node in $\tree{\prefix{S_d}{L-1}}$ will never leave this subgraph. To see this, consider some node $v \in V$ and assume that the link $(v,d)$ is failed. Let this node be the $i$-th node in  $\bl{\prefix{S_d}{L-1}}$. This node $v$ has only edges that reach into $\cl{\prefix{S_d}{L-2}}{i}$ and $W$. Only in case $v$ forwards flows to some node in $\cl{\prefix{S_d}{L-2}}{i}$, they would leave $\tree{\prefix{S_d}{L-1}}$. However, even in such a case the flow must travel back to $\tree{\prefix{S_d}{L-1}}$ to reach the destination. As $v$ is the only node in $\tree{\prefix{S_d}{L-1}}$ that can be reached from nodes in $\cl{\prefix{S_d}{L-2}}{i}$ -- and forwarding the packet back to $v$ would cause a forwarding loops -- this requires multiple hops and does not optimize for shortest paths.
    This description can be visualized with the help of $\cref{fig:fixed-level}$. Imagine $\prefix{S_d}{L-1} = S \circ (k/2)$,  $\prefix{S_d}{L-2} = S$ and $v$ is the node in $\bl{S \circ (k/2)}$ with the purple square. Then, the only nodes $v$ can reach outside of $\tree{\prefix{S_d}{L-1}} = \tree{S \circ (k/2)}$ lie in $\cl{S}{2}$. However, these nodes only have a single connection into $\tree{\prefix{S_d}{L-1}}$. Therefore, as $\mathcal{P}$ attempts to route over shortest paths, $v$ must always forward over nodes in $w \in W$ in case the edge $(v,d)$ is failed. This implies that flows arriving at some node in $\tree{\prefix{S_d}{L-1}}$ will never leave this bipartite graph. 
    
    In summary, when only failing edges inside $\tree{\prefix{S_d}{L-1}}$, our problem can be reduced to finding a set of failed edges in a complete bipartite graph, where each node in the set $V$ starts with $\Theta(k^{L-1})$ many flows. For routing purposes, the protocol $\mathcal{P}'$ is employed, which contains the routing distributions of $\mathcal{P}$ in $\tree{\prefix{S_d}{L-1}}$. By \cref{lem:bipartite-lower-bound} of the bipartite graph analysis, we have that -- in expectation -- some node $v^*\in V$ receives load from $\BigO{\log k / \log \log k}$ other nodes in $V$. As $\mathcal{P}$ is fairly balanced, this implies that this node $v^*$ receives $\Omega
(k^{L-1} \cdot  \log k / \log \log k)$ load.
 \end{proof}

\atreefairlybalanced*

\begin{proof} 
    We start by showing that $\mathcal{P}_C$ is fairly balanced. To that end, we consider a fixed destination $d$ and assume that $\mathcal{F} = \emptyset$, i.e., no edges are failed. The following statement shows that every node which is not part of $\bl{\prefix{S_d}{\ell}}$ for any $1 \leq \ell \leq L$ receives $\BigO{\polylog k}$ load.
    
    \begin{observation}
        Assume that $\mathcal{F} = \emptyset$ and consider an arbitrary block $\bl{S}$ such that, either (i) $S$ is \emph{not} a prefix of $S_d$, or (ii) $S = \emptyseq$. Then, for a fixed node $v \in \bl{S}$, it holds that $\mathcal{L}(v) \leq \polylog k$ with probability at least $\geq 1 - k^{-10}$.
    \end{observation}
    \begin{proof}
         We fix such a node $v$ in a block $\bl{S}$, that fulfills either (i) or (ii). Nodes might receive flows from his neighbors due to either \ruleone or \ruletwo of \cref{def:interval-routing}. We argue that $v$ only receives flows according to \ruleone. In case assumption (ii) holds, which means $S=\emptyseq$, this is easy to see as $\bl{\emptyseq}$ has no parent. Otherwise, in case (i) holds, consider $\bl{S^P}$ -- the parent block of $\bl{S}$. Even if $S^P$ is a prefix of $S_d$, the no neighbor of $v$ in $\bl{S^P}$ will forward its flows to $v$ according to \ruletwo. This is because no edges are failed, which allows the nodes in $\bl{S^P}$ to forward their load directly to some $d_{|S|,i}$, $i\geq 0$. These nodes all lie in a block $\bl{S'}$, where $S'$ is a prefix of $S_d$. Therefore, our node $v$ will only receive flows from neighbors due to \ruleone. Let $v$ be in cluster interval $\cli{S}{i}{j}$ for some $i,j\geq 1$. It will only (potentially) receive flows from neighbors $v_1$ that are part of $\vcli{S}{i}{j-1 \mod K}$. As no edges are failed and failover edges are chosen uniformly, it follows that $v_1$ will forward its flows to $v$ with probability $1 / I = \BigO{\log k / k}$. As failover entries are chosen independently from other nodes, and $|\vcli{S}{i}{j-1 \mod K}| \leq k$, it follows from a Chernoff bound application that at most $\BigO{\log k}$ nodes in $\vcli{S}{i}{j}$ will forward their load to $v$ \whp. Note that these $\BigO{\log k}$ nodes lie in a block $\bl{S \circ i}$, for some $1 \leq i \leq k/2$. As $S$ was is not a prefix of $S_d$, it follows that also $S \circ i$ is not a prefix of $S_d$. Hence, we can repeat this argument, which implies that each of these $\BigO{\log k}$ nodes again only receives load from $\BigO{\log k}$ of its neighbors, which lie on level $|S| + 2$. This way, we can upper-bound the load $v$ receives with the help of a tree of degree $\BigO{\log k}$. This tree has depth at most $L - |S| = \BigO{1}$ as at this point our argument arrives at the highest level, which receive flows from no other node. Therefore, $v$ receives $\BigO{\log^L k} = \polylog k$ flows in total. By increasing the constant hidden in the $O$-notation we can easily achieve a probability of $1-k^{-10}$ for this event.
    \end{proof}
    What remains to show is that, \whp, nodes in $\bl{\prefix{S_d}{\ell}}$ for $1 \leq \ell \leq L$ receive $\Theta(k^{\ell})$ load. To that end, we first consider the load of nodes in $\cl{\emptyseq}{i}$. The statement of \cref{obs:a-tree-1}, which is used in the inductive step from $\ell = 0 \rightarrow \ell = 1$ of \cref{lem:a-tree-induction} and implies that the nodes in any cluster $\vcl{\emptyseq}{i} \setminus \{d_{1,i}\}$ receives a load of $\BigO{k}$ can easily be adapted to also include a matching lower bound of $\Omega(k)$. The reason for this is that the proof relies on a sequence of Chernoff bound applications, which can be applied to bound both the upper and lower tail. As the load of nodes in $\vcl{\emptyseq}{i} \setminus \{d_{1,i}\}$ is forwarded to $\cl{\emptyseq}{i}$ according to routing \ruleone, this implies that, for any $i\geq1$, $\sum_{v \in \cl{\emptyseq}{i}} \mathcal{L}(v) = \Theta(k)$. We now present the following result.
    \begin{observation}
    \label{obs:fairly-balanced}
        Let $1 \leq \ell < L$. Then, for any node $d_{\ell,i}$ in $\bl{\prefix{S_d}{\ell}}$ it holds \whp, that 
        \[
            \mathcal{L}(d_{\ell,i}) = \BigO{k} + \sum_{v \in \cl{\prefix{S_d}{\ell-1}}{i}} \mathcal{L}(v).
        \]
    \end{observation}
    \begin{proof}
    Fix some  $1 \leq \ell < L$  and $d_{\ell,i}$, the $i$-th node in $\bl{\prefix{S_d}{\ell}}$. Assume that this node lies in the cluster $\cl{\prefix{S_d}{\ell}}{i'}$ for some $i' \geq 0$. Consider now the nodes that might forward their flows to $v$ in a single hop. The reason to be forwarded to $v$ might either be \ruleone or \ruletwo of \cref{def:interval-routing}. In case of \ruleone, the flow must originate from some node in $\vcl{\prefix{S_d}{\ell}}{i'} \setminus \{d_{\ell+1,i'}\}$. Note that we can exclude $d_{\ell+1,i'}$ here as it is part of $\bl{\prefix{S_d}{\ell+1}}$ and does not forward flows according to \ruleone. In case of \ruletwo this will be \emph{all} nodes of $\cl{\prefix{S_d}{\ell - 1}}{i}$ -- as we assume no edges to be failed all of them may forward directly to $d_{\ell,i}$. According to \cref{obs:a-tree-1}, we can bound the number of flows forwarded via nodes corresponding to \ruleone by $\BigO{k}$ \whp, which yields that
    \[
        \mathcal{L}(d_{\ell,i}) = \BigO{k} +  \sum_{v \in \cl{\prefix{S_d}{\ell}}{i-1}} \mathcal{L}(v) \qedhere
    \]
    \end{proof}
    We are now ready to show that, \whp, any node in $\bl{\prefix{S_d}{\ell}}$ receives $\Theta(k^{\ell})$ load, $1 \leq \ell \leq L$. We start with $\ell = 1$. From \cref{obs:fairly-balanced} together with the fact that $\sum_{v \in \cl{\emptyseq}{i}} \mathcal{L}(v) = \Theta(k)$, we immediately get that any node in $\bl{\prefix{S_d}{1}}$ has load $\BigO{k} + \Theta(k) = \Theta(k)$ as desired. This can be translated into an induction, where any node in $\bl{\prefix{S_d}{\ell}}$
    receives 
    \[
        \BigO{k} + \sum_{v \in \cl{\prefix{S_d}{\ell}}{i}} \mathcal{L} (v) = \BigO{k} +  (k/2) \cdot \Theta(k^{\ell}) = \Theta(k^{\ell + 1})
    \]
    load. In the second step we used the induction hypothesis, which states that any node in $\bl{\prefix{S_d}{\ell}}$ has $\Theta(k^{\ell}$ load. Note that this especially includes nodes in $\cl{\prefix{S_d}{\ell}}{i}$ for any $i\geq 1$ which are a subset of $\bl{\prefix{S_d}{\ell}}$. Therefore, the protocol $\mathcal{P}_C$ is fairly balanced.
    
    \medskip
    
    Next, we show that $\mathcal{P}_C$ only forwards over shortest paths. Consider a packet with destination $d$, which lies on some node in  $\bl{S}$ of level $L$, and assume for now that no edges are failed. In \cref{sec:clos-definition} we established that our topology is a fat-tree, and when compressing the nodes in each blocks into a single node, then the resulting graph forms a tree. This makes it easy to see, which routing strategy results in a shortest path to $d$. That is, first the packet need to be send until the first block $\bl{S'}$ is reached such that $S'$ is a prefix of $S_d$. At this point, the destination $d$ is in the subtree rooted at $\bl{S'}$. Assume that $S' = \prefix{S_d}{\ell}$, i.e., that $S'$ is prefix of length $\ell$ of $d$. To reach $d$, the packet must be forwarded to the child $\bl{S_d}{\ell + 1}$ to further approach the destination. All other child blocks of $\bl{S'}$ lie in a different subtree of $d$. Summarizing, to achieve a shortest path in the absence of failures, the packet needs to be routed from child to parent for blocks $\bl{S}$ such that $S$ is not a prefix of $S_d$. Otherwise, it needs to be forwarded to the child block, whose sequence is length $|S|+ 1$ prefix of $S_d$.
    
    It is easy to see that the protocol in \cref{def:interval-routing} achieves exactly this routing behavior in absence of failures. Now, consider a fixed node $u \in \bl{S}$ that receives a packet with destination $d$, and assume the edges to some of its neighbors are failed. Let $\mathcal{F}_u$ with $|\mathcal{F}_u| < I / 2$ denotes this set. In case of $S$ \emph{not} being a prefix of $S_d$, the protocol $\mathcal{P}_C$ follows \ruleone, and as only $I/2$ edges are failed, the adversary cannot prevent the packet from being forwarded to $\bl{S}$. As sketched above, this implies that the packet still follows a shortest path.
    However, if $u \in \bl{S}$ such that $S$ is a length $\ell$ prefix of $S_d$, then the packet is forwarded to some node according to \ruletwo.  Let $\cl{S}{i}$ denote the cluster of block $\bl{S}$ in which $u$ resides. If $\mathcal{F}_u$ does not contain $d_{|S|+1,i}$, then the packet is forwarded to $d_{|S|+1,i}$ in $\bl{\prefix{S_d}{\ell+1}}$, which is in line with the shortest path routing strategy above. In case the edge $(u,d_{|S|+1,i})$ is failed, this is impossible. However, this is also reflected in best achievable the shortest path in $(V,E \setminus \mathcal{F}_v)$: In order to reach the destination, the packet needs to be forwarded to some other node $u'$ in $\bl{S}$ with $u' \neq u$ so that $u'$ can then forward the packet towards $\bl{\prefix{S_d}{\ell+1}}$. As $u$ and $u'$ lie in the same block, this is impossible to achieve in less than $2$ hops. In such case any shortest path must run over some node $w \in \bl{S \circ i}$ for some $i\geq 1$, or in other words some child of $\bl{S}$. From $w$ the packet then needs to move back into $\bl{S}$ to some $u' \neq u$. Note that the desired behavior of forwarding the packet from $u$ to such a node $w$ is captured by  $\mathcal{P}_C$ as $u$ follows \ruletwo.
\end{proof}

\subsection{Concentration Inequalities}

\begin{theorem}[McDiarmid's Inequality, cf.~\cite{MRT12}]
\label{thm:mcdiarmid}
    Let $X_1, X_2, ... ,X_m \in \mathcal{X}^m$ be a set of $m\geq 1$ independent random variables and assume that there exist $c_1, ... , c_m >0$ such that $f: \mathcal{X}^m \rightarrow \mathrm{R}$ satisfies the following conditions:
    \[
        f(x_1, ..., x_i, ... ,x_m) - f(x_1, ... , x_i', ..., x_m)| \leq c_i,
    \]
    for all $i \in [m]$ and any points $x_1, ..., x_m, x_i' \in \mathcal{X}$. Let $Y$ denote $f(X_1, ... , X_m)$, then for all $t > 0$, the following inequalities hold:
    \[
        \Pr \left[Y - \Ex{Y} \geq t \right] ~,~ \Pr \left[Y - \Ex{Y} \leq -t \right] \leq \exp \left(-2t^2 / \sum_{i=1}^{m} c_i^2 \right).
    \]
\end{theorem}

\begin{theorem}[Chernoff Bound]
\label{thm:chernoff}
    Let $X_1,X_2, ... ,X_n$ be independent Poisson trials such that $\Pr[X_i = 1] = p_i$ and $X =\sum_{i=1}^{n} X_i$. Then, if $\mu = \Ex{X} = \sum_{i=1}^n p_i$ and $\mu_l \leq \mu \leq  \mu_u$, we have for any $\delta \in [0,1]$ that
    \[
        \Pr \left[X \geq (1+\delta) \mu_u \right] \leq \exp(-\mu_u \cdot \delta^2 / 3)
    \]
    \[
        \Pr \left[X  \leq (1-\delta) \mu_l \right]  \leq \exp(-\mu_l \cdot \delta^2 / 2).
    \]
\end{theorem}
\section{Future Work}
\label{sec:conclusion}

While we provided almost tight bounds on the achievable
congestion with randomized local fast rerouting, 
our work leaves several interesting avenues for future research.
In particular, it will be interesting to generalize our lower bound
and to study
the achievable resilience and congestion under arbitrary 
traffic patterns. 
It would also be interesting to generalize our results
to AB fat-trees~\cite{liu2013f10}: our Clos topology 
essentially corresponds to an ``A fat-tree''.

\bibliographystyle{plainurl}
\bibliography{paper}

\end{document}